\RequirePackage[l2tabu, orthodox]{nag}
\documentclass[10pt,a4paper,onecolumn]{article}
\usepackage[margin=2cm]{geometry}
\usepackage{amsmath,amssymb,amsfonts,amsthm,mathtools,thmtools,hyperref,graphicx,cite}
\usepackage[dvipsnames]{xcolor}
\usepackage[nottoc]{tocbibind}
\usepackage[capitalize]{cleveref}
\usepackage{bbm}
\hypersetup{citecolor=violet,linkcolor=blue,urlcolor=violet,colorlinks=true,pdfborder={0 0 0},linktocpage=true}

\graphicspath{{./figures/}}

\newcommand{\bra}[1]      {\langle #1|}
\newcommand{\ket}[1]      {|#1\rangle}
\newcommand{\braket}[2]   {\left\langle #1\middle|#2\right\rangle}
\newcommand{\braopket}[3] {\left\langle #1\middle|#2\middle|#3\right\rangle}
\newcommand{\comm}[2]     {\left[ #1 , #2 \right]}
\newcommand{\twicomm}[2]  {\comm{#1}{#2}_\alpha}

\newcommand{\abs}[1]      {\left|#1\right|}
\newcommand{\trivert}     {{\vert\kern-0.4ex\vert\kern-0.4ex\vert}}
\newcommand{\trivertbig}  {{\big\vert\kern-0.4ex\big\vert\kern-0.4ex\big\vert}}
\newcommand{\trivertBig}  {{\Big\vert\kern-0.4ex\Big\vert\kern-0.4ex\Big\vert}}
\newcommand{\norm}[1] 	  {\left\|#1\right\|}
\newcommand{\normF}[1] 	  {\left\|#1\right\|_{F}}
\newcommand{\normpk}[1]   {\left\|#1\right\|_{(p,k)}}
\newcommand{\normu}[1]    {{\left\vert\kern-0.4ex\left\vert\kern-0.4ex\left\vert #1 
		\right\vert\kern-0.4ex\right\vert\kern-0.4ex\right\vert}}

\theoremstyle{plain}
\newtheorem{thm}{Theorem} 
\newtheorem{thm2}{Theorem} 
\newtheorem{appthm}{Theorem} 
\newtheorem{lem}{Lemma}
\newtheorem{corr}[thm]{Corollary}
\newtheorem{defn}{Definition}

\numberwithin{lem}{section}
\numberwithin{appthm}{section}

\newcommand{\Pic}{\bar{\Pi}}
\newcommand{\round}[1]{\left\lfloor#1\right\rceil}

\let\Re\relax
\DeclareMathOperator{\Re}{Re}
\DeclareMathOperator{\Tr}{Tr}
\DeclareMathOperator{\Span}{Span}

\usepackage[normalem]{ulem}

\newcommand*{\STF}[1]{}
\newcommand*{\CTC}[1]{}
\newcommand*{\CTCsout}[1]{}
\newcommand*{\ctc}[1]{}
\newcommand*{\stf}[1]{}
\newcommand*{\stfsout}[1]{}
\newcommand*{\add}[1]{{#1}}
\newcommand*{\del}[1]{}

\makeatletter
\def\lastcountervalue#1#2{\expandafter\xdef\csname l@#1@value\endcsname{#2}}
\def\savelastvalue#1{\AtEndDocument{\immediate\write\@auxout{\string\lastcountervalue{#1}{\arabic{#1}}}}}
\def\getlastvalue#1{\@ifundefined{l@#1@value}{??}{\@nameuse{l@#1@value}}}
\AtEndDocument{\def\lastcountervalue#1#2{\def\reserved@a{#2}\expandafter\ifx\csname l@#1@value\endcsname\reserved@a\else\@tempswatrue\fi}}
\makeatother

\newif\ifcompile\compilefalse
\newif\ifcomments\commentsfalse
\newif\ifchanges\changesfalse

\ifcompile
	\usepackage{pgfplots}
	\usetikzlibrary{external,arrows,calc,decorations.markings,through}
	\usepackage{ifthen}
	\pgfplotsset{compat=newest}
	\newlength{\figheight}   \newlength{\figwidth}
\fi

\ifcomments
	\usepackage{framed}
	\newcounter{stfcomments}
	\newcounter{ctccomments}
	\newcounter{totcomments}
	
	\renewcommand*{\STF}[1]{
		\protect\stepcounter{stfcomments}\protect\stepcounter{totcomments}{{\bfseries\color{red}\centering\begin{framed}STF: #1\end{framed}}}}
	\renewcommand*{\CTC}[1]{
		\protect\stepcounter{ctccomments}\protect\stepcounter{totcomments}{{\bfseries\color{blue}\centering\begin{framed}CTC: #1\end{framed}}}}
	\renewcommand*{\CTCsout}[1]{
		{{\bfseries\color{blue}\centering\begin{framed}CTC: \sout{#1}\end{framed}}}}
	\renewcommand{\stf}[1]{\protect\stepcounter{stfcomments}\protect\stepcounter{totcomments}{\color{red}\text{stf:} #1}}
	\renewcommand{\stfsout}[1]{{\color{red}{\sout{stf: #1}}}}
	\renewcommand{\ctc}[1]{{\protect\stepcounter{ctccomments}\protect\stepcounter{totcomments}\color{blue}ctc: #1} }
	
	\savelastvalue{stfcomments}
	\savelastvalue{ctccomments}
	\savelastvalue{totcomments}
\fi

\ifchanges
	\renewcommand*{\add}[1]{{\color{OliveGreen}#1\color{black}}}
	\renewcommand*{\del}[1]{{\color{red}#1\color{black}}}
\fi

\title{Approximate symmetries of Hamiltonians}
\author{Christopher T.\ Chubb$^1$~~and~~Steven T.\ Flammia$^{1,2}$\\
	{\small\itshape $^{1}$Centre for Engineered Quantum Systems, University of Sydney, Sydney, NSW, Australia.}\\
	\add{{\small\itshape $^{2}$Center for Theoretical Physics, Massachusetts Institute of Technology, Cambridge, MA, USA.}}}
\begin{document}
\maketitle

\begin{abstract}
We explore the relationship between approximate symmetries of a gapped Hamiltonian and the structure of its ground space. 
We start by considering approximate symmetry operators, defined as unitary operators whose commutators with the Hamiltonian have norms that are sufficiently small. 
We show that when approximate symmetry operators can be restricted to the ground space while approximately preserving certain mutual commutation relations. 
We generalize the Stone-von Neumann theorem to matrices that \emph{approximately} satisfy the canonical (Heisenberg-Weyl-type) commutation relations, and use this to show that approximate symmetry operators can certify the degeneracy of the ground space even though they only approximately form a group. 
Importantly, the notions of ``approximate'' and ``small'' are all independent of the dimension of the ambient Hilbert space, and depend only on the degeneracy in the ground space. 
Our analysis additionally holds for any gapped band of sufficiently small width in the excited spectrum of the Hamiltonian, and we discuss applications of these ideas to topological quantum phases of matter and topological quantum error correcting codes. 
Finally, in our analysis we also provide an exponential improvement upon bounds concerning the existence of shared approximate eigenvectors of approximately commuting operators\add{ under an added normality constraint,} which may be of independent interest.
\end{abstract}



\ifcomments
	\begin{framed}
		\vspace{-.5cm}
		\begin{center}
			\Large
			STF:~\getlastvalue{stfcomments}\hspace{2cm}
			CTC:~\getlastvalue{ctccomments}\hspace{2cm}
			Total:~\getlastvalue{totcomments}
		\end{center}
		\vspace{-.5cm}
	\end{framed}
\fi

\hypersetup{linkcolor=black}
\tableofcontents
\hypersetup{linkcolor=red}


\section{Introduction}

Given a quantum system described by a Hamiltonian $H$, a symmetry is simply an 
operator that \del{exactly }commutes with $H$. \del{Often symmetries are 
restricted to be hermitian or unitary. $H$ can be simultaneously block 
diagonalized with the symmetry operator}\add{The symmetry can be block 
diagonalized with respect to the energy eigenspaces}, and \add{so }the 
degeneracy within these blocks 
is \add{constrained}\del{determined} by the symmetry. In a system that 
possesses exact symmetries, a 
sufficiently weak perturbation will preserve the number of states of any band 
gapped away from the rest of the spectrum, but the symmetries will generally 
become only approximate. 

In this work we consider a natural converse to this: suppose we know that a system has some approximate symmetries and a gapped band, such as the ground space band. Can we ``unperturb'' the symmetries into exact symmetries within the given band? Can we also use the approximate group structure of the approximate symmetries to count the degeneracy within the band? We answer these questions in the affirmative, giving quantitive bounds on when such a procedure can be performed, and thus when such approximate symmetries can be used as certificates of ground space degeneracy.

A related area of mathematical research with a long and rich history is the 
relationship between the properties of approximately and exactly commuting 
matrices. An exemplary problem which dates back as far as the 
1950s~\cite{Rosenthal1969, Halmos1976, Szarek1990, Davidson1985, Berg1991} is 
whether a pair of approximately commuting matrices lie near an exactly 
commuting pair, i.e.\ whether there exists a dimension independent $\delta>0$ 
for each $\epsilon>0$ such that \add{for all $H,S$ with $\norm{H}\!,\norm{S}\leq 1$,}
\[ \norm{\comm{H}{S}}\leq \delta\qquad \implies \qquad \exists 
\tilde{H},\tilde{S}:\bigl[\tilde{H},\tilde{S}\bigr]=0,\text{ where } 
\bigl\|H-\tilde{H}\bigr\|,\bigl\|S-\tilde{S}\bigr\|\leq \epsilon, 
\]
\add{where here and throughout the norm $\| \cdot \|$ is the operator norm.} 
 Interpreting $H$ as the Hamiltonian and $S$ as\del{the action of} a symmetry, this problem can be interpreted as whether approximate symmetries are necessarily near exact symmetries of a perturbed system. 
It has been shown that just such a theorem holds if all matrices are 
Hermitian~\cite{Lin1997, Kachkovskiy2014, Friis1996, hastings09b}. 
A\del{nother} physical consequence of this is that a pair \add{approximately 
}commuting observables can be approximately simultaneously 
measured~\cite{hastings09b}. 

For \del{discrete symmetries, described by }unitary matrices\del, the above is \add{however }known to be generally false~\cite{Voiculescu1983}. This is due to a K-theoretic obstruction~\cite{Exel1989, Exel1991, Choi1988}, though it is true if this obstruction vanishes~\cite{Lin1997,Kachkovskiy2014, Lin1998}, or under the assumption of a spectral gap~\cite{Osborne2008}. Imposing a form of self-duality analogous to time-reversal symmetry the relevant K-theoretic obstruction reduces to the spin Chern number of a fermionic system~\cite{Loring2013}, highlighting a link between the fields of topologically ordered quantum systems~\cite{Wen2012} and approximately commuting matrices.

Here we will consider Hamiltonians $H$ with multiple non-commuting approximate 
symmetries, and establish a connection to the ground space degeneracy. 
\add{Ground space degeneracy is a property of a quantum system that plays a 
special role in several important applications, such as quantum coding theory 
and the study of phases of matter. Quantum codes, especially those encoded into 
ground spaces of local Hamiltonians, are the leading candidates for 
thermally stable quantum memories~\cite{Preskill1999,Brown2016}; in these 
models approximate symmetries constitute approximate logical operators and the 
ground space degeneracy corresponds to the code size. In the context of 
condensed matter systems, the link between symmetries and degeneracies plays an 
important role both in classical symmetry-breaking 
phases~\cite{LandauLifshitz}, and also in exotic quantum phases, such as those 
exhibited by topologically ordered models~\cite{Wen2012}. 
Unfortunately, determining the ground space degeneracy of (finite) systems is generally 
\#\textsf{P}-complete, even for gapped bands~\cite{Brown2011}. 
However, if we restrict to more structured examples such as 1D-local spin 
systems, then ground spaces can in fact be efficiently approximated~\cite{ChubbFlammia2015,Huang2014}.
Our results show that when structure is present in the form of certain mutual commutation relations, one can obtain certifiable bounds on the degeneracy of a ground space by only knowing bounds on these relations. 
We go into more detail about these two applications in \cref{sec:localapplications}. }

The form of non-commutation we will consider will involve \emph{twisted} commutation relations.

\begin{defn}[Twisted commutator]
	For $\alpha\in[0,1)$, the \emph{twisted commutator} is defined as 
	\[\twicomm{X}{Y}:=XY-e^{2i\pi\alpha}YX.\]
	We will refer to $\alpha$ as the \emph{twisting parameter}, and for some 
	unitarily invariant norm $\normu{\cdot}$ we will refer to 
	$\normu{\twicomm{\cdot}{\cdot}}$ as the \emph{twisted commutation value}. 
	When considering a pair of operators in tandem such that each has a small 
	twisted commutation value we will refer to it as a \emph{twisted pair}. 
\end{defn}
We \add{note}\del{remark} that the $\alpha=0$ and $\alpha=1/2$ cases correspond to the commutator and anti-commutator respectively.

Commuting operators exist in all dimensions, finite or infinite. Twisted commuting operators on finite-dimensional spaces, however, only exist in certain dimensions depending on the twisting parameter, e.g.\ no $\alpha\neq 0$ twisted commutator can non-trivially vanish in a one-dimensional space. For general operators, twisted commuting operators were studied in some detail in Ref.~\cite{Yang2004}. If we restrict to unitary operators however, the Stone-von Neumann Theorem\footnote{As usually stated, the Stone-von Neumann theorem is much more general than \cref{thm:svnt}. We will only be concerned with twisted commutation in finite-dimensional spaces, and unconcerned with uniqueness, so this form will suffice for our purposes.}~\cite
{Hall2013,Rosenberg2004} classifies the dimensions in which twisted commutation can occur. \del{In this paper we will generalize this connection into the regime of \emph{approximately} twisted commuting operators.}

\begin{thm}[Finite-dimensional Stone-von Neumann theorem]
	\label{thm:svnt}
	Given $\alpha=p/q$ with $p,q$ coprime, then unitary operators $X$ and $Y$ which exactly twisted commute as
	\[ \twicomm{X}{Y}=0 \]
	only exist in dimensions which are multiples of $q$.
\end{thm}

\add{In this paper we will generalize this connection into the regime of \emph{approximately} twisted-commuting operators. Properties of both approximate commuting, and approximately twisted-commuting operators are reviewed in Ref.~\cite{Said2014}. The rigidity of algebraic structures to small perturbations in the commutation relations that define them has been studied in several other settings, such as the soft torus~\cite{softtorus,softtorus2} and approximate representations of groups~\cite{Babai,Friedl,Moore2010}.}

Suppose we have a physical system with a \add{self-adjoint }Hamiltonian $H$\add{, 
acting on a possibly infinite-dimensional Hilbert space}. Let $\Pi$ be the 
orthogonal projector onto the \add{finite-dimensional }ground space, and 
$\Pic:=I-\Pi$. For simplicity, take the ground state energy of $H$ to be zero, 
such that $\Pi H= 0$. As well as this, we will assume that the excited states 
are gapped away from the ground space, such that they all have an energy at 
least $\Delta$, i.e. $H\geq \Delta \Pic$. For such a system there exist two 
notions of symmetry we will discuss.

\begin{defn}[Symmetry]
	\label{defn:symmetry}
	\del{Let $U$ be a unitary operator. We will refer to it as a \emph{ground symmetry} if it commutes with the ground space projector}
	\add{We define a \emph{ground symmetry} as an operator $U$ that commutes 
	with the ground space projector}
	\[ \comm{U}{\Pi}=0, \]
	\add{and acts unitarily on the ground space $\Pi U^\dag U\Pi=\Pi UU^\dag \Pi=\Pi$. Moreover, we refer to a unitary} 
	\del{and }as an \emph{$\epsilon$-approximate symmetry} if it approximately commutes with the Hamiltonian with respect to a given unitarily invariant norm
	\[ \normu{\comm{U}{H}}\leq \epsilon. \]
\end{defn}

\add{Here we use $\normu{\cdot}$ to denote any unitarily invariant norm. }

The error thresholds we are going to consider will depend on the spectral gap $\Delta$ of the system in question. One way to improve the scaling with the gap would be to consider symmetries defined not by commutation with the Hamiltonian, but by commutation with functions of the Hamiltonian. For example we could consider commutation with an (unnormalized) Gibbs state
\[ \normu{\comm{U}{e^{-\beta H}}}\leq \epsilon. \]
Such a symmetry can be seen to be an $\epsilon$-approximate symmetry of $H'=I-e^{-\beta H}$, which shares a ground space with $H$ and has a gap of $1-e^{-\beta \Delta}$. If we have some control over the temperature, such as in Monte Carlo simulations, then this gives a tradeoff we can use to improve the gap scaling. If for example we set $\beta = \ln (2)/\Delta$, then we get a fixed gap of $1/2$. A similar analysis could be performed with any function of $H$ which leaves the relevant band gapped.

\subsection{Results}
\label{subec:results}

\add{The main goal of this paper will be to establish a connection between twisted commuting symmetries and the ground space dimension, even when the relevant commutation relations are only approximate. A key feature of our bounds is that they can be expressed entirely in terms of the Hamiltonian, and do not require objects such as the ground space projector, which can often be prohibitively difficult to calculate, represent, or perform calculations with. Without access to the ground space projector, whether or not a unitary is a ground symmetry cannot be directly verified.} 
	
In \cref{sec:inexact} we will explore the relationship between approximate and ground symmetries, showing that an approximate symmetry is always near a ground symmetry. Extending this to the case of multiple symmetries, we will see that approximate symmetries can be restricted to the ground space with low distortion, implying the existence of \add{unitaries}\del{operators} on the ground space with certain twisted commutation relations. 
\add{In showing these results, we will make repeated use of the following function and note some simple bounds on it,
\begin{align*}
	f:[0,1] \to [0,1]\,, \qquad f(x) := 1-\sqrt{1-x}\,, \qquad \tfrac{x}{2} \le f(x) \le x \,. 
\end{align*}
Then our first main result is the following.}

\begin{thm}[Restriction to the ground space]
	\label{thm:restriction}
	For two $\epsilon$-approximate symmetries $U$ and $V$ which approximately twisted commute
	\[ \normu{\twicomm{U}{V}}\leq \delta, \]
	\add{then if $\xi:=\epsilon/\Delta<1$ }there exists unitaries $u$ and $v$ acting on the ground space which also approximately twisted commute as
	\add{\begin{align*}
		\normu{\twicomm{u}{v}}&\leq \delta+2\xi^2 +4f(\xi^2)\,.
	\end{align*}}
\end{thm}

\add{Rather importantly, we note that the above theorem holds independent of the ground space dimension. This will allow us to use approximate symmetries alone as witnesses of ground space degeneracy, circumventing the need for direct access to the ground space, which is often inaccessible in non-exactly solvable models.}

Note that for simplicity we will henceforth take the band in consideration to 
be an exactly degenerate ground space. We will see however that our proof will 
rely not on the bound $H\geq \Delta\Pic$, but on its relaxation $H^2\geq 
\Delta^2\Pic$, meaning that the band could be anywhere in the spectrum, so long 
as it is gapped on both sides by at least $\Delta$. Furthermore we can take 
$w:=\normu{H\Pi} \ge 0$ when our band has a potentially non-zero width. By 
considering the new Hamiltonian $H':=H-H\Pi$, we get that our restricted result 
holds for more general bands once the necessary changes have been made.

\begin{corr}[Restriction to a general band]
	\label{corr:band}
	If there are two $\epsilon$-approximate symmetries $U$ and $V$ which approximately twisted commute
	\[ \normu{\twicomm{U}{V}}\leq \delta, \]
	\add{then if $\xi':=(\epsilon+w)/\Delta<1$ }there exists unitaries $u$ and $v$ acting on band of gap $\Delta$ and width $w$ which also approximately twisted commute
	\add{\begin{align*}
		\normu{\twicomm{u}{v}}&\leq \delta+2\xi'^2 +4f(\xi'^2)\,.
	\end{align*}}
\end{corr}

Now that we have restricted our symmetries down to the ground space, by studying the relationship between dimensionality and approximate twisted commutation, we can hope to use these twisted symmetries as witnesses of ground space degeneracy. As above, we will henceforth adhere to the convention of upper case letters denoting operators which act on the system as a whole, and lower case operators which only act on the ground space.


In \cref{sec:collection} we start by giving a proof of \cref{thm:svnt}, and consider generalizing this argument to the case of \emph{approximately} twisted commuting operators.
We consider a twisted pair of unitaries, and construct states which can be used to lower bound the number of eigenvalues these operators possess. By doing so we will show that if these operators have a sufficiently small twisted commutation value in the operator norm, then a lower bound on their degeneracy can be inferred.

\begin{thm}
	\label{thm:single}
	If $u$ and $v$ are unitaries such that for some $d\in\mathbb{N}$
	\[ \norm{\comm{u}{v}_{1/d}}<\frac{2}{d-1}\Bigl[1-\cos\pi/d\Bigr],  \]
	then the dimension of each operator is at least $d$.
\end{thm}

While we do not have a closed form bound on the twisted commutation value required to certify other dimensions ($d\neq 1/\alpha$), in \cref{app:lowerbound} we discuss an algorithm to determine which degeneracies are certified by twisted pairs of given parameters. Using this we will plot the dimension that can be certified as a function of both the twisting parameter and the corresponding twisted commutation value.

In \cref{app:approxeigen} we strengthen existing results on shared approximate eigenvectors for approximately commuting operators when a normality condition is introduced, exponentially improving the dimension dependence of the bounds relative to known results~\cite{Bernstein1971}. Using this, in \cref{subsec:double} we consider extending this procedure to the case of two pairs of twisted commuting unitaries. Here we will once again construct a set of ground states, showing that for sufficient parameters that they are linearly independent. Using this we can obtain a similar dimensionality lower bound.

\begin{thm}
	\label{thm:double}If $u_1$, $u_2$, $v_1$ and $v_2$ are unitaries such that 
	they satisfy the commutation relations
	\[ \norm{\comm{u_1}{u_2}}\leq\gamma \qquad \norm{\comm{u_1}{v_2}}\leq\delta \qquad \norm{\comm{u_2}{v_1}}\leq\delta\]
	and twisted commutation relations
	\[ \norm{\comm{u_1}{v_1}_{1/d_1}}\leq\delta \qquad \norm{\comm{u_2}{v_2}_{1/d_2}}\leq\delta \]
	with $d_1\leq d_2$ and
	\[ \sqrt{\gamma}d_1d_2+(d_1+d_2)\delta< \frac{\sin^2(\pi/2d_1)}{(d_1d_2-1)^2}, \]
	then the dimension of each operator is at least $d_1d_2$.
\end{thm}


In \cref{sec:minimum} we provide a more comprehensive analysis for the case of a single twisted pair. Leveraging results from spectral perturbation theory, we find an explicit closed form for the minimum twisted commutation value\add{ for a class of norms known as  the $(p,k)$-Schatten-Ky Fan norms.  These are defined as the $p$-norm of the largest $k$ singular values, or more formally as 
\begin{align*}
	\normpk{X}:=\sup_A\left\lbrace(\Tr\abs{AX}^p)^{1/p}\,\middle|\,\norm{A}\leq1,\,\mathrm{rank}(A)\leq k\right\rbrace\,.	
\end{align*}
For a $g$-dimensional operator, the special case $k=g$ reduces to the Schatten $p$-norm, and the case $p=1$ reduces to the Ky Fan $k$-norm. 
In particular, the $p=\infty$ and $(p,k)=(2,g)$ special cases reduce to the operator and Frobenius norms respectively.}

\begin{thm}[Minimum twisted commutation value]
	\label{thm:mintcv}
	Suppose that $u$ and $v$ are $g$-dimensional unitaries, \add{then for any $p\geq2$ the twisted commutator is lower bounded 
	\[ \Bigl\|\twicomm{u}{v}\Bigr\|_{(p,k)}\geq 2k^{1/p}\sin\left(\pi \abs{\frac{\round{g\alpha}-g\alpha}{g}}\right), \]
	where $\normpk{\cdot}$ is the $(p,k)$-Schatten-Ky Fan norm.} 
	\del{then for any $p\geq 2$ we have
	\[ \Bigl\|\twicomm{u}{v}\Bigr\|_{p}\geq 2g^{1/p}\sin\left(\pi \abs{\frac{\round{g\alpha}-g\alpha}{g}}\right). \]}
	Moreover this bound is tight, in that sense that there exist \add{families of }$g$-dimensional unitaries which saturate the above bound\add{s} and only depend on $\round{g\alpha}$, the nearest integer to $g\alpha$.
\end{thm}

For a given twisted pair, all dimensions for which the twisted commutation value falls below this minimum can therefore be ruled out as valid dimensions. As this bound is not monotonic as a function of $g$, it not only provides a lower bound, but a full classification of which dimensions are disallowed.

After giving proofs of the main results outlined above in Sections \ref{sec:inexact}, \ref{sec:collection}, and \ref{sec:minimum}, we turn to broader discussion and applications of these ideas. Section \ref{sec:localapplications} is devoted to discussion of future directions for this work that add the additional constraint that the Hamiltonian is \emph{local}, and we discuss the relationship to the notions of topological order and topological quantum codes. In particular we show how recent numerical methods for studying quantum many-body systems~\cite{BridgemanFlammiaPoulin2016} could leverage the bounds presented here to provide certificates of the topological degeneracy of certain quantum systems.


\section{Restriction to the ground space}
\label{sec:inexact}

In this section we will make precise the notion that approximate \add{symmetries }can be utilized as proxies of ground symmetries. We first establish a relationship between approximate symmetries and the ground symmetries that they imply. Then we consider operators with approximate twisted commutation relations, and we show that these can also be restricted faithfully to the ground space with low distortion. 


Constructing a ground symmetry from an approximate symmetry will come in two steps. First we will pinch the symmetry $U$ with respect to $\Pi$, giving an operator $P$ for which $\comm{P}{\Pi}=0$. While this will render $P$ no longer unitary, we will see that \add{its action upon the ground space }\del{it }will still be \emph{approximately} unitary. \del{We will then show that approximate unitaries are nearly unitary. }Using this we will construct a nearby \del{unitary}\add{operator} $\tilde{U}$ that retains commutation with the ground space projector\add{, and acts unitarily on the ground space}, thus constituting a ground symmetry.

\add{We will start by showing that the off-diagonal blocks of an approximate symmetry are small, and then follow by showing that its action on the ground space is approximately unitary.}
\del{To bound the change in our symmetry that is caused by pinching, we will start by bounding the off-diagonal blocks of $U$ with respect to $\Pi$.}

\begin{lem}[Small off-diagonal blocks]
	\label{lem:off-diag}
	If $U$ is an $\epsilon$-approximate symmetry, 
	then off-diagonal blocks of $U$ with respect to $\Pi$ have bounded norms, 
	in particular \del{$\normu{\Pic U \Pi} \le \epsilon/\Delta$ and $\normu{\Pi U \Pic} \le \epsilon/\Delta$}\add{$\normu{\Pic U\Pi +\Pi U \Pic}\leq \epsilon/\Delta$}. \add{For Hamiltonians of the form $H=\Delta\Pic$, this inequality is tight.} 
\end{lem}
\begin{proof}
	\add{We start by noting that $\abs{A}^2\geq \abs{B}^2$ implies 
	$\abs{AX}^2\geq \abs{BX}^2$ for any $X$, where $\abs{M}:=\sqrt{M^\dag M}$. 
	Taking $X$ 
	to be finite-rank, we have from Weyl's inequalities~\cite{Weyl} that the singular values of $AX$ majorize 
	those of $BX$. Unitarily invariant norms\footnote{Following 
	Ref.~\cite{Bhatia,Schatten} we adopt the normalization $\normu{A}=\norm{A}$ 
	for all rank-1 operators $A$.} act as symmetric gauge functions on 
	finite-rank operators\cite{vonNeumann,Schatten,Bhatia}, which 
	implies from Refs.~\cite[Prop.IV.1.1, 
	Thm.IV.2.2]{Bhatia} that $\normu{AX}\geq \normu{BX}$---a similar 
	argument for the adjoint also gives $\normu{XA}\geq \normu{XB}$. 
	Because $H$ has a gapped band with projector $\Pi$, we have that $H^2\geq \Delta^2 \Pic$. 
	Using this, we can bound the off-diagonal blocks in terms of the commutator
	\begin{align*}
	\Delta \normu{\Pic U\Pi+\Pi U\Pic}
	&=\normu{(\Pic\Delta +\Pi)\left(\Pic U\Pi +\Pi U\Pic\right)\left(\Delta \Pic +\Pi\right)}\\
	&\stackrel{\mathrm{~}}{\leq} \normu{(\Pic H +\Pi)\left(\Pic U\Pi +\Pi U\Pic\right)\left(H \Pic +\Pi\right)}\\
	&=\normu{\Pic HU\Pi+\Pi UH\Pic}\,,
	\end{align*}
	where the inequality follows from the aforementioned monotonicity property. 
	Now using the unitary invariance of the norm (since $\Pic - \Pi$ is unitary), we find that 
	\begin{align*}
	\normu{\Pic HU\Pi+\Pi UH\Pic}&=\normu{\Pic HU\Pi-\Pi UH\Pic}\\
	&=\normu{\Pic \comm{H}{U}\Pi+\Pi \comm{H}{U}\Pic}\\
	&\stackrel{\mathrm{~}}{\leq}\normu{\comm{H}{U}}\leq \epsilon,
	\end{align*}
	where the second equality makes use of $H\Pi = 0$ and the first inequality 
	is the pinching inequality.}
	
	\add{With regard to tightness, if we take $H=\Delta\Pic$ then we can see $\comm{H}{U}$ has no on-diagonal blocks, and therefore $\Delta\left(\Pic U\Pi-\Pi U \Pic\right)=\comm{H}{U}$. Taking norms of both sides of this equation give $\Delta\normu{\Pic U\Pi+\Pi U \Pic}=\normu{\comm{H}{U}}$, meaning that $\normu{\Pic U\Pi+\Pi U \Pic}\leq \epsilon/\Delta$ is tight.}
		
	\del{
	The Hamiltonian $H$ is a positive operator, and we have in particular $H\geq\Delta\Pic\geq 0$, which in turn implies that $H^2\geq \Delta^2\Pic$. Then the modulus of the off-diagonal block of $U$ is bounded in the semidefinite operator ordering as follows:
	\begin{align*}
		\Delta^2 \abs{\Pic U \Pi}^2
		&={(\Delta\Pic U \Pi)^\dag(\Delta\Pic U \Pi)}\\
		&={\Pi U^\dag\cdot\Delta^2\Pic\cdot U \Pi}\\
		&\leq {\Pi U^\dag\cdot H^2\cdot U \Pi}\\
		&= {(H U \Pi)^\dag(H U \Pi)}\\
		&=\abs{HU\Pi}^2.
	\end{align*}
	By the min-max theorem, this ordering on the moduli implies that the singular values of $\Delta\Pic U\Pi$ are, listed in order, each upper bounded by those of $HU\Pi$. As unitarily invariant norms are monotonic functions of the singular values, this allows us to conclude that $\Delta \normu{\Pic U \Pi}\leq \normu{HU\Pi}$. Using $H\Pi=0$ we can bound this in turn by the commutator:
		\begin{align*}
			\Delta\cdot\normu{\Pic U \Pi\vphantom{\big(}}
			&\leq\normu{H U \Pi\vphantom{\big(}}\\
			&=\normu{\big.\!\left(H U -UH\right) \Pi}\\
			&\leq\normu{\big.\!\comm{H}{U}\Pi}\\
			&\leq\normu{\big.\!\comm{H}{U}}\\
			&\leq \epsilon
		\end{align*}
	An analogous argument can be used for $\Pi U \Pic$.}
\end{proof}

\add{
	\begin{lem}[Approximate unitarity on the ground space]
		\label{lem:approx unit}
		For an $\epsilon$-approximate symmetry $U$ with 
		$\xi:=\epsilon/\Delta\leq1$, the action on the ground space is 
		approximately unitary
		\begin{align*}
		\normu{\Pi-\abs{\Pi U\Pi}}\leq f(\xi^2),
		\end{align*}
		where $f(x):=1-\sqrt{1-x}$. In the operator norm, this expression is tight.
\end{lem}	
}\begin{proof}
		\add{First we can bound $\abs{\Pi U\Pi}^2$ near $\Pi$ by using the unitarity of $U$ itself as
		\begin{align*}
			\Pi -\abs{\Pi U\Pi}^2
			&=\Pi-\Pi U^\dag \Pi U\Pi\\
			&=\Pi U^\dag U\Pi -\Pi U^\dag \Pi U\Pi\\
			&=\Pi U^\dag\Pic U\Pi\\
			&=\abs{\Pic U \Pi}^2.
		\end{align*}
		Together with \cref{lem:off-diag}, the sub-multiplicativity of 
		unitarily invariant norms on finite-rank operators let us conclude that 
		$\trivert{\Pi-\abs{\Pi U\Pi}^2\trivert}\leq \xi^2$. Next we need to 
		use this bound on $\trivert{\Pi -\abs{\Pi U \Pi}^2}\trivert$, and create a bound on $\normu{\Pi-\abs{\Pi U\Pi}}$.}
		
		\add{Consider a function $f(x)=\sum_{n=1}^{\infty}a_nx^n$, where 
		$a_n>0$ and $f(1)<\infty$. For any finite-rank operator $0\leq X\leq 
		1$, we can use the triangle inequality and submultiplicativity of 
		$\normu{\cdot}$ to derive a Jensen-like inequality
		\begin{align*}
			\normu{f(X)}
			=\normu{\sum_{n=1}^{\infty}a_n X^n}
			\leq \sum_{n=1}^{\infty}a_n\normu{X^n}
			\leq \sum_{n=1}^{\infty}a_n\normu{X}^n
			=f\left(\normu{X}\right).
		\end{align*}
		If we let $a_n=\Gamma(n-1/2)/2\sqrt{\pi}n!$, then we get $f(x)=1-\sqrt{1-x}$ for $x\in [0,1]$. If we let $X=\Pi-\abs{\Pi U\Pi}$, then applying the above gives
		\begin{align*}
			\normu{\Pi -\abs{\Pi U\Pi}}
			&=\normu{f\left(\Pi -\abs{\Pi U\Pi}^2\right)}\\
			&\leq f\left(\normu{\Pi-\abs{\Pi U\Pi}^2}\right)\\
			&\leq f(\xi^2).		
		\end{align*} 
		We note that $x/2\leq f(x)\leq x$, which means that this bound improves upon the bound trivially given by the contractivity of $\Pi U\Pi$,
		\begin{align*}
			\normu{\Pi -\abs{\Pi U\Pi}}\leq \normu{\Pi-\abs{\Pi U\Pi}^2}\leq 
			\xi^2.
		\end{align*}
		}
		
		\add{For the purposes of tightness, consider a two-dimensional Hilbert space, and a Hamiltonian $H$ and unitary $U$ given by 
		\begin{align*}
			H=\begin{pmatrix}
			0 & \\ & \Delta
			\end{pmatrix}
			\qquad\text{and}\qquad
			U=\begin{pmatrix}
			\cos\phi & \sin\phi \\-\sin\phi & \cos\phi
			\end{pmatrix}.
		\end{align*}
		In the operator norm $\norm{\comm{U}{H}}=\Delta\sin\phi$ and $\norm{\Pi -\abs{\Pi U\Pi}}=1-\cos\phi$, which saturates the above bound.
		}
\end{proof}

Using \add{these bounds}\del{this} we can now construct a ground symmetry $\tilde{U}$ by pinching $U$ with respect to $\Pi$, \add{and then restoring unitarity on the ground space}\del{and enforcing unitarity}.

\begin{lem}[Approximate symmetries are nearly ground symmetries]
	\label{lem:as=n-les}
	For an $\epsilon$-approximate symmetry $U$ with 
	$\xi:=\epsilon/\Delta\leq1$, 
	there exists a ground symmetry $\tilde{U}$ \add{which is close to 
	$U$}\del{such that}
	\[ \normu{U-\tilde{U}}\leq \del{6\epsilon/\Delta}\add \xi\add{+f(\xi^2)}, \]
	\add{and closer still in the ground space
	\[  \normu{\Pi\left(U-\tilde{U}\right)\Pi}\leq f(\xi^2)\,.\]
	}\add{The first inequality is tight to leading 
	order in $\xi$, and the second is tight in the operator norm.}
\end{lem}
\begin{proof}
	\del{Firstly, pinch the operator $U$ into the diagonal blocks $P:=\Pi U \Pi +\Pic U\Pic$. Clearly $\comm{P}{\Pi}=0$ by construction, but $P$ is no longer unitarity. Using \cref{lem:off-diag} we can see that $P$ is close to $U$.
	\begin{align*}
		\normu{U-P}
		=\normu{\Pi U \Pic+\Pic U\Pi}
		\leq\del{\normu{\Pi U \Pic}+\normu{\Pic U\Pi}
		\leq 2}\epsilon/\Delta
	\end{align*}}
	
	\del{We note in passing that $\abs{U-P}$ is block diagonal, so if $\normu{\cdot}$ were a Schatten $p$-norm, then using the identity $\norm{A\oplus B}^p_p=\norm{A}^p_p+\norm{B}^p
		-p$, the factor of $2$ acquired here can be strengthened to $2^{1/p}$, disappearing entirely in the operator norm.}
	\del{The operator $P$ is contractive by the pinching lemma, so we can use the fact that $U$ is unitary to show that $P$ is approximately unitary.
	\begin{align*}
		\normu{I-PP^\dag}
		&=\normu{UU^\dag-PP^\dag}\\
		&=\normu{(U-P)P^\dagger+U(U^\dag -P^\dag)}\\
		&\leq\normu{U-P}\cdot\norm{P}+\normu{U-P}\\
		&\leq2\normu{U-P}\\
		&\leq 4\epsilon/\Delta \,.
	\end{align*}
	Next we will see that approximate unitaries are near unitaries. To see this, take the singular value decomposition $P=LSR^\dag$, and let $\tilde{U}:=LR^\dag$. By the unitary invariance of the norm, the distance between $P$ and $\tilde{U}$ is entirely controlled by the singular values $S$. By non-negativity of singular values and contractivity of $P$, we have the operator inequality $0\leq S\leq 1$, which implies $S^2 \leq S$ and thus $I-S\leq I-S^2$. Using this we can see that $\tilde{U}$ is close to $P$,
	\begin{align*}
		\normu{\tilde{U}-P}
		&=\normu{L(I-S)R^\dag}\\
		&=\normu{\phantom{\big.}I-S}\\
		&\leq\normu{I-S^2}\\
		&=\normu{L(I-S^2)L^\dag}\\
		&=\normu{I-PP^\dag}\\
		&\leq 4\epsilon/\Delta\,.
	\end{align*}
	It should be noted that their exists an ordering on the singular values, corresponding to performing the SVD block-wise, such that the block structure of $P$ with respect to $\Pi$ is retained in $\tilde{U}$, and so $[\tilde{U},\Pi]=0$ forming a ground symmetry as desired. }
	
	\add{We start by considering the polar of decompositon $\Pi U \Pi=W\abs{\Pi 
	U\Pi}$. As the ground space $\mathrm{im}(\Pi)$ is an invariant subspace of 
	$\Pi U\Pi$, we can take\footnote{\add{Such a $W$ could be found by 
	performing the polar decomposition restricted to the ground space, and 
	padding the unitary out to act as the identity on the rest of the space.}} 
	$W$ to also leave the ground space invariant, $\comm{W}{\Pi}=0$. Given 
	this, we define our ground symmetry to be $\tilde{U}:=\Pi W\Pi+\Pic U\Pic$.}
	
	\add{We will now consider bounding the distance between $U$ and $\tilde U$ block-wise. The off-diagonal blocks are bounded by \cref{lem:off-diag} as
	\[ \normu{\Pi\left(\tilde{U}-U\right)\Pic+\Pic\left(\tilde{U}-U\right)\Pi}
	=\normu{\Pic U\Pi+\Pic U\Pi}\leq \xi.\]
	The bound on the ground space however follows from \cref{lem:approx unit}
	\[ \normu{\Pi\left(\tilde U -U\right)\Pi}=\normu{W\Pi-W\abs{\Pi 
	U\Pi}}=\normu{\Pi-\abs{\Pi U\Pi}}\leq f(\xi^2). \]
	Finally the fact that $U$ was unchanged on the excited space trivially implies
	\[ \normu{\Pic \left(\tilde U-U\right)\Pic }=0. \]
	Putting everything together, this gives the desired bound
	\begin{align*}
		\normu{\tilde{U}-U}
		&=\normu{\Pic\left(\tilde{U}-U\right)\Pi
		+\Pi\left(\tilde{U}-U\right)\Pic
		+\Pi\left(\tilde{U}-U\right)\Pi
		+\Pic\left(\tilde{U}-U\right)\Pic}\\
		&\leq\normu{\Pic\left(\tilde{U}-U\right)\Pi
			+\Pi\left(\tilde{U}-U\right)\Pic}
			+\normu{\Pi\left(\tilde{U}-U\right)\Pi}
			+\normu{\Pic\left(\tilde{U}-U\right)\Pic}\\
		&\leq\xi+f(\xi^2).
	\end{align*}}
	\add{
	As for tightness, \cref{lem:off-diag} gives that $\normu{\Pic 
	U\Pi+\Pi U\Pic}=\normu{\comm{H}{U}}/\Delta$ for Hamiltonians of the form 
	$H=\Delta\Pic$. If we assume that $\normu{\comm{U}{H}}=\epsilon$, then 
	applying the pinching inequality gives
	\begin{align*}
		\normu{U-\tilde{U}}\geq 
		\normu{\Pic\left(U-\tilde{U}\right)\Pi+\Pi\left(U-\tilde{U}\right)\Pic}=\normu{\Pic
		 U\Pi +\Pi U\Pic}=\xi,
	\end{align*}
	which proves our bound on $\normu{U-\tilde{U}}$ is tight to leading order 
	in $\xi$. The tightness of the norm distance in the ground space follows 
	directly from the tightness of \cref{lem:approx unit}.} 
\end{proof}

%
%
%

We will now consider how the existence of nearby ground symmetries allows twisted commutation relations of approximate symmetries to be pulled down into the ground space. 



\stepcounter{thm2}
\begin{thm2}[Restriction to the ground \add{space}\del{band}]
		For two $\epsilon$-approximate symmetries $U$ and $V$ which 
		approximately twisted commute
		\[ \normu{\twicomm{U}{V}}\leq \delta, \]
		\add{then if $\xi:=\epsilon/\Delta\leq1$ }there exists unitaries $u$ 
		and 
		$v$ acting on the ground space which also approximately twisted commute 
		as
		\add{\begin{align*}
			\normu{\twicomm{u}{v}}&\leq \delta+2\xi^2 +4f(\xi^2)\,.
		\end{align*}}
\end{thm2}

\begin{proof}
	Consider a $\tilde{U}$ and $\tilde{V}$ given by applying \cref{lem:as=n-les} to $U$ and $V$ respectively, such that
	\del{\[ \trivert U-\tilde{U}\trivert,\,\trivert V-\tilde{V}\trivert\leq \frac{\epsilon}{\Delta}. \]}
	\add{\[ \normu{\Pi U\Pi-\Pi\tilde U\Pi},\, \normu{\Pi V\Pi-\Pi\tilde 
	V\Pi}\leq f(\xi^2). \]}
	\del{Using this we can see that the twisted commutator cannot grow much:
	\begin{align*}
		\normu{[\tilde{U},\tilde{V}]_{\alpha}-[{U},{V}]_{\alpha}}
		&=\normu{\tilde{U}\tilde{V}-e^{2\pi i\alpha}\tilde{V}\tilde{U}-{U}{V}+e^{2\pi i\alpha}{V}{U}}\\
		&=\trivertBig 
		\left[\left(\tilde{U}-U\right)V+\tilde{U}\left(\tilde{V}-V\right)\right]
		+e^{2\pi i\alpha}\left[\left(\tilde{V}-V\right)U+\tilde{V}\left(\tilde{U}-U\right)\right]\trivertBig\\
		&\leq
		\normu{\left(\tilde{U}-U\right)V}
		+\normu{\tilde{U}\left(\tilde{V}-V\right)}
		+\normu{\left(\tilde{V}-V\right)U}
		+\normu{\tilde{V}\left(\tilde{U}-U\right)}\\
		&=
		2\normu{\tilde{U}-U}
		+2\normu{\tilde{V}-V}\\
		&\leq 24\frac{\epsilon}{\Delta}
	\end{align*}}
	\add{\!\!\!\!\!Next we consider the twisted commutator of $U$ and $V$, and 
	that of 
		$\tilde{U}$ and $\tilde V$, both projected into the ground space. By 
		expanding out the twisted commutators we have
	\begin{align*}
		\Pi\twicomm{U}{V}\Pi-\Pi\twicomm{\tilde U}{\tilde V}\Pi
		&=\twicomm{\Pi U\Pi}{\Pi V \Pi}-\twicomm{\Pi \tilde U \Pi}{\Pi\tilde V\Pi}+\Pi U\Pic\cdot\Pic V\Pi-e^{2\pi i\alpha}\Pi V\Pic\cdot\Pic U\Pi,\\
		&=\left(\Pi U\Pi-\Pi \tilde U\Pi\right)\cdot \Pi V\Pi-e^{2\pi i\alpha}\Pi \tilde V\Pi\cdot\left(\Pi U\Pi-\Pi \tilde U\Pi\right)\\
		&~\quad+\Pi \tilde U\Pi\cdot\left(\Pi V\Pi-\Pi \tilde V\Pi\right)-e^{2\pi i\alpha}\left(\Pi V\Pi-\Pi \tilde V\Pi\right)\cdot \Pi U\Pi\\
		&~\quad+\Pi U\Pic\cdot\Pic V\Pi-e^{2\pi i\alpha}\Pi V\Pic\cdot\Pic U\Pi.
	\end{align*}
	Using the triangle inequality, the contractivity of $\Pi U\Pi$ and $\Pi V\Pi$, and the bound on the off-diagonal blocks from \cref{lem:off-diag}, we can bound this as required:
	\begin{align*}
		\normu{\Pi\twicomm{U}{V}\Pi-\Pi\twicomm{\tilde U}{\tilde V}\Pi}
		&\leq \normu{\left(\Pi U\Pi-\Pi \tilde U\Pi\right)\cdot \Pi V\Pi}+\normu{\Pi \tilde V\Pi\cdot\left(\Pi U\Pi-\Pi \tilde U\Pi\right)}\\
		&~\qquad+\normu{\Pi \tilde U\Pi\cdot\left(\Pi V\Pi-\Pi \tilde V\Pi\right)}+\normu{\left(\Pi V\Pi-\Pi \tilde V\Pi\right)\cdot \Pi U\Pi}\\
		&~\qquad+\normu{\Pi U\Pic\cdot\Pic V\Pi}+\normu{\Pi V\Pic\cdot\Pic U\Pi}\\
		&\leq 2\normu{\Pi U\Pi-\Pi \tilde U\Pi}+2\normu{\Pi V\Pi-\Pi \tilde V\Pi}\\
		&~\qquad+\normu{\Pi U\Pic}\cdot\normu{\Pic V\Pi}+\normu{\Pi V\Pic}\cdot\normu{\Pic U\Pi}\\
		&\leq 4f(\xi^2)+2\xi^2.
	\end{align*}
	}
	
	\del{So we conclude
	\[ \normu{[\tilde{U},\tilde{V}]_{\alpha}}\leq \delta+24\frac{\epsilon}{\Delta}. \]
	Next let $u$ and $v$ be the ground space restrictions of $\tilde{U}$ and $\tilde{V}$. As each is a ground symmetry, $u$ and $v$ must be unitary. The twisted commutator value of the ground symmetries also carries through to these restrictions:
	\begin{align*}
		\normu{\twicomm{u}{v}}
		&=\normu{\twicomm{u\oplus 0}{v\oplus0}}\\
		&=\normu{[\Pi \tilde{U}\Pi,\Pi \tilde{V}\Pi]_{\alpha}}\\
		&=\normu{\Pi[\tilde{U},\tilde{V}]_{\alpha}}\\
		&\leq\normu{[\tilde{U},\tilde{V}]_{\alpha}}\\
		&\leq\delta+24\frac{\epsilon}{\Delta}. 
	\end{align*}}

	\add{Next we let $u$ and $v$ be the restriction of $\tilde{U}$ and $\tilde 
	V$ to the ground space respectively. As each are ground symmetries, $u$ and 
	$v$ are both unitaries. If we consider the embedding of operators on the 
	ground space back into the larger Hilbert space, then we can use the above 
	to bound the twisted commutator of our ground space unitaries
		\begin{align*}
			\normu{\twicomm{u}{v}}
			&=\normu{\twicomm{u\oplus 0}{v\oplus0}}\\
			&=\normu{[\Pi \tilde{U}\Pi,\Pi \tilde{V}\Pi]_{\alpha}}\\
			&=\normu{\Pi[\tilde{U},\tilde{V}]_{\alpha}\Pi}\\
			&\leq\normu{\Pi[U,V]_{\alpha}\Pi}+2\xi^2+4f(\xi^2)\\
			&\leq\normu{[U,V]_{\alpha}}+2\xi^2+4f(\xi^2)\\
			&\leq\delta+2\xi^2+4f(\xi^2).
		\end{align*}
	}
	Note that if we had a set of more than two unitaries, this additive growth in the twisted commutation value would hold equally for every pair separately.
%
\end{proof}

\stepcounter{thm2}


\section{Degeneracy lower bounds}
\label{sec:collection}

In this section we show how twisted pairs of unitary operators can be used to give lower bounds on the degeneracy of the ground space. 
We start by considering an exact twisted pair and the Stone-von Neumann theorem. 
We will then show how this argument can be generalized to approximate twisted pairs, and how a lower bound on the degeneracy follows from an upper bound on the twisted commutator value. Finally we will see how this can also be extended to more general collections of twisted commuting operators through the example case of two twisted pairs that are approximately mutually commuting. 

\subsection{Stone-von Neumann Theorem}
\label{subsec:exact}


Consider a $u$ and $v$ which exactly twisted commute, so that
$ uv=e^{2i\pi \alpha} vu$. 
Let $(\lambda,\ket{\psi})$ be an eigenpair of $u$. Using the twisted commutation relation, we see that $\ket{\psi'}:=v\ket{\psi}$ forms a $\lambda e^{2i\pi \alpha}$-eigenvector.
It follows that $v$ forms an isomorphism between the $\lambda$ and $\lambda e^{2i\pi \alpha}$-eigenspaces of $u$, which allows us to conclude that their dimensions must be the same. 
Carrying this argument forward, we can see that any eigenspaces whose eigenvalues differ by any power of $e^{2i\pi \alpha}$ must also be isomorphic.


Suppose we take $\alpha\in \mathbb{Q}$, with $\alpha=p/q$ with $p$, $q$ coprime. As we can see in \cref{fig:exact}, a simple divisibility argument implies that the eigenspaces come in isomorphic multiples of $q$, which therefore implies that the overall dimension of $u$ and $v$ is a multiple of $q$ also.

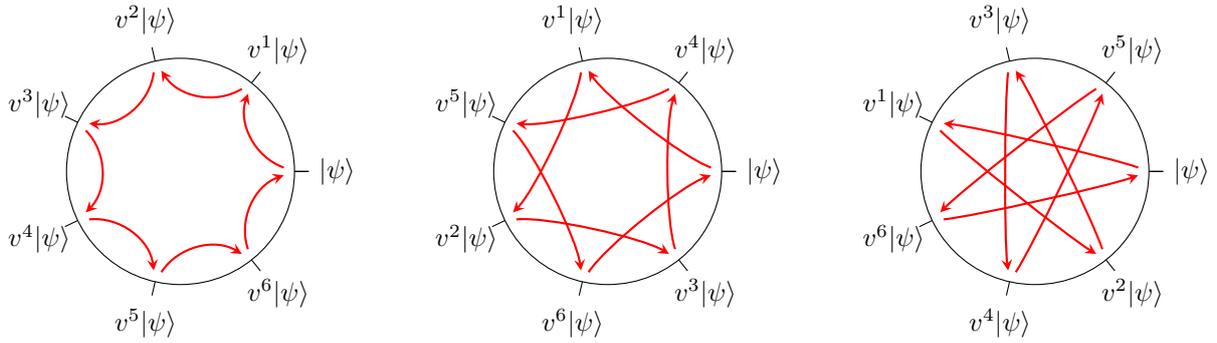
\begin{figure}[t!]
	\def\d{7}
	\ifcompile
	\def\R{2}
	\def\r{0.25}
	\def\t{0}
	\def\a{-20}
	\def\ddr{.2}
	\def\c{.75}
	\def\s{.75}
	\centering
	\def\skip{1}
	\tikzsetnextfilename{1skip}
	\begin{tikzpicture}[scale=.75]
	\pgfmathsetmacro{\half}{(\d-1)/2} 
	\draw (0,0) circle (\R);
	\draw (\t:\R) -- (\t:\R+\r);
	\node at (\t:\R+\s) {$\ket{\psi}$};
	\foreach \x in {1,...,\d} {
		\draw (\t+360*\x/\d:\R) -- (\t+360*\x/\d:\R+\r);
		\draw[-stealth,red,thick] ($(\t+360*\x/\d-\skip*360/\d:\R)-(\a+\t+360*\x/\d-\skip*360/\d:\ddr)$)
		.. controls ($(\t+360*\x/\d-\skip*360/\d:\R)-(\a+\t+360*\x/\d-\skip*360/\d:\c)$) and ($(\t+360*\x/\d:\R)-(-\a+\t+360*\x/\d:\c)$) .. ($(\t+360*\x/\d:\R)-(-\a+\t+360*\x/\d:\ddr)$);	
		\ifthenelse{\equal{\x}{\d}}{}{\node at (\t+360*\x/\d:\R+\s) {$v^{\pgfmathparse{int(\x)}\pgfmathresult}\ket{\psi}$};}
	}
	\end{tikzpicture}\qquad
	\def\skip{2}
	\tikzsetnextfilename{2skip}
	\begin{tikzpicture}[scale=.75]
	\pgfmathsetmacro{\half}{(\d-1)/2} 
	\draw (0,0) circle (\R);
	\draw (\t:\R) -- (\t:\R+\r);
	\node at (\t:\R+\s) {$\ket{\psi}$};
	\foreach \x in {1,...,\d} {
		\draw (\t+360*\skip*\x/\d:\R) -- (\t+360*\skip*\x/\d:\R+\r);
		\draw[-stealth,red,thick] ($(\t+360*\skip*\x/\d-\skip*360/\d:\R)-(\a+\t+360*\skip*\x/\d-\skip*360/\d:\ddr)$)
		.. controls ($(\t+360*\skip*\x/\d-\skip*360/\d:\R)-(\a+\t+360*\skip*\x/\d-\skip*360/\d:\c)$) and ($(\t+360*\skip*\x/\d:\R)-(-\a+\t+360*\skip*\x/\d:\c)$) .. ($(\t+360*\skip*\x/\d:\R)-(-\a+\t+360*\skip*\x/\d:\ddr)$);	
		\ifthenelse{\equal{\x}{\d}}{}{\node at (\t+360*\skip*\x/\d:\R+\s) {$v^{\pgfmathparse{int(\x)}\pgfmathresult}\ket{\psi}$};}
	}
	\end{tikzpicture}\qquad
	\def\skip{3}
	\tikzsetnextfilename{3skip}
	\begin{tikzpicture}[scale=.75]
	\pgfmathsetmacro{\half}{(\d-1)/2} 
	\draw (0,0) circle (\R);
	\draw (\t:\R) -- (\t:\R+\r);
	\node at (\t:\R+\s) {$\ket{\psi}$};
	\foreach \x in {1,...,\d} {
		\draw (\t+360*\skip*\x/\d:\R) -- (\t+360*\skip*\x/\d:\R+\r);
		\draw[-stealth,red,thick] ($(\t+360*\skip*\x/\d-\skip*360/\d:\R)-(\a+\t+360*\skip*\x/\d-\skip*360/\d:\ddr)$)
		.. controls ($(\t+360*\skip*\x/\d-\skip*360/\d:\R)-(\a+\t+360*\skip*\x/\d-\skip*360/\d:\c)$) and ($(\t+360*\skip*\x/\d:\R)-(-\a+\t+360*\skip*\x/\d:\c)$) .. ($(\t+360*\skip*\x/\d:\R)-(-\a+\t+360*\skip*\x/\d:\ddr)$);	
		\ifthenelse{\equal{\x}{\d}}{}{\node at (\t+360*\skip*\x/\d:\R+\s) {$v^{\pgfmathparse{int(\x)}\pgfmathresult}\ket{\psi}$};}
	}
	\end{tikzpicture}
	\else
		\centering
		\includegraphics{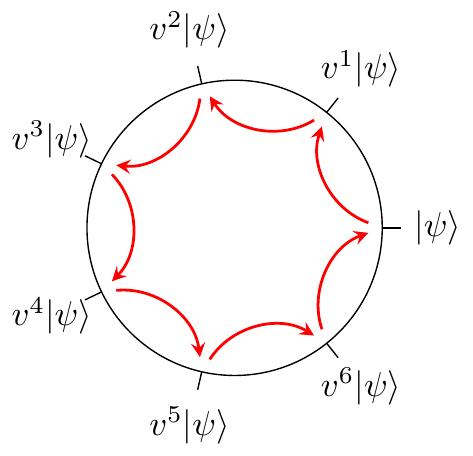}
		\qquad
		\includegraphics{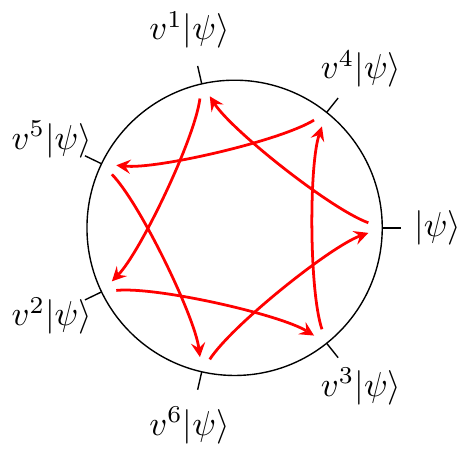}
		\qquad
		\includegraphics{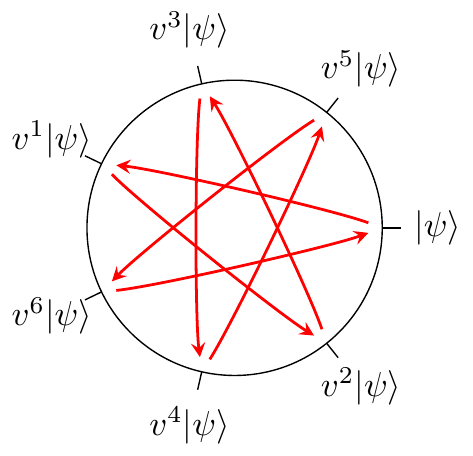}
	\fi
	\caption{The action of powers of $v$ on an eigenvector $\ket{\psi}$ of $u$. On the left $\comm{u}{v}_{1/\d}=0$, in the centre $\comm{u}{v}_{2/\d}=0$, and on the right $\comm{u}{v}_{3/\d}=0$. Here the position of the circle represents the corresponding eigenvalue of $u$.}
	\label{fig:exact}
\end{figure}

We now generalize this connection between the twisted commutator and the spectrum of one of the operators to allow for only approximate twisted commutation.

\subsection{One twisted pair}

\label{subsec:single}

Let us first extend the above argument to the case of a single approximate twisted pair. 
For simplicity, we consider the case where $\alpha = p/q$ with $p=1$ and $q=d$, so the corresponding phase in the twisted commutator is $\eta:=e^{2i\pi/d}$. 
This is not much of a restriction since if $p>1$ we can replace $v$ with $v^{\bar{p}}$ where $\bar{p}$ is the modular multiplicative inverse of $p$ such that $\bar{p} p = 1 \bmod q$ and then apply the results of the $p=1$ case. 
Under this substitution the twisted commutator will grow by at most a factor of $\lfloor q/2 \rfloor$. 
However, in Appendix~\ref{app:lowerbound} we will show an alternative method that in fact works for arbitrary $\alpha \in \mathbb{R}$ and gives tighter bounds than this simple reduction. We also consider without loss of generality the case where $u$ has at least one $+1$ eigenvalue, which can always be achieved by redefining $u$ by multiplying by a complex unit phase factor.

Suppose we have two unitaries $u$ and $v$ such that
\[ \norm{\comm{u}{v}_{1/d}}=\norm{uv-\eta vu}\leq \delta. \]
Our results will show that these operators must, for sufficiently small $\delta$, be at least $d$-dimensional. To do this we will explicitly show that $u$ has at least $d$ distinct eigenvalues.

Let $\ket{\psi}$ be a $+1$ eigenvector of $u$, \add{i.e.\ }\del{such that }$u\ket{\psi}=\ket{\psi}$. Consider the orbit of $\ket{\psi}$ under $v$, i.e.\ the states $\ket{j}:=v^j\ket{\psi}$ for $j=-\lfloor \frac{d-1}{2}\rfloor,\dots,\lceil\frac{d-1}{2}\rceil$. These vectors are precisely the vectors depicted in Figure~\ref{fig:exact}. We first show that these are approximate eigenstates of $u$. 


\begin{lem}[Change in expectation value: One pair]
	\label{lem:t-col:change1}
	The expectation value of $u$ with respect to $\ket{j}$ is approximately $\eta^j$, specifically
	\[ \abs{\braopket{j}{u}{j}-\eta^j}\leq \abs{j}\delta. \]
\end{lem}
\begin{proof}
	This follows from the twisted commutator of $u$ and $v$ being small. By expanding the commutator and applying the triangle inequality we can see that $\norm{uv-\eta vu}\leq \delta$ implies $\norm{uv^j-\eta^jv^ju}\leq \abs{j}\delta$.
	\del{\begin{align*}
		\norm{uv^j-\eta^jv^ju}
		&=\norm{\Big.\!\left(uv^j-\eta vuv^{j-1}\right)
				+\left(\eta vuv^{j-1}-\eta^2v^2uv^{j-2}\right)+\cdots
				+\left(\eta^{j-1}v^{j-1}uv-\eta^jv^ju\right)}\\
		&\leq\norm{uv^j-\eta vuv^{j-1}}
			+\norm{\eta vuv^{j-1}-\eta^2v^2uv^{j-2}}+\cdots
			+\norm{\eta^{j-1}v^{j-1}uv-\eta^jv^ju}\\
		&=\norm{\left(uv-\eta vu\right)v^{j-1}}
			+\norm{\eta v\left(uv-\eta vu\right)v^{j-2}}+\cdots
			+\norm{\eta^{j-1} v^{j-1}\left(uv-\eta vu\right)}\\
		&= \abs{j}\cdot\norm{uv-\eta vu}
	\end{align*}}
	From this we can see that the expectation value of $\ket{j}$ lies close to $\eta^j$:
	\begin{align*}
		\abs{j}\delta
		&\geq \norm{uv^j-\eta ^jv^ju}\\
		&= \norm{v^{-j}uv^j-\eta ^ju}\\
		&\geq \abs{\braopket{\psi}{\left[v^{-j}uv^j-\eta ^ju\right]}{\psi}}\\
		&\geq \abs{\braopket{\psi}{v^{-j}uv^j}{\psi}-\eta^j\braopket{\psi}{u}{\psi}}\\
		&=\abs{\braopket{j}{u}{j}-\eta^j}.
	\end{align*}
\end{proof}

So we can see that the $\lbrace \ket{j}\rbrace$ form a set of vectors with expectation values distributed approximately evenly around the unit circle, much like the states in the $\delta=0$ case as seen in \cref{fig:exact}. To relate these states to the dimensions of $u$ and $v$, we will now show that there must exist an eigenvalue of $u$ near the expectation value of each state.

\begin{lem}[Existence of eigenvalues]
	\label{lem:t-col:eigen1}
	If there exists a state $\ket{x}$ such that
	\[ \abs{\braopket{x}{u}{x}-e^{i\theta}}\leq \zeta \]
	then $u$ possesses a nearby eigenvalue $e^{i\phi}$ such that
	\[ \abs{\phi-\theta}\leq \cos^{-1}(1-\zeta). \]
\end{lem}
\begin{proof}
	The bound on the expectation value with respect to $u$ implies
	\[ \Re \braopket{x}{e^{-i\theta}u}{x}\geq 1-\zeta. \]
	As this expectation value is a convex combination of the eigenvalues of $u$, all of which lie on the unit circle, there must exists an eigenvalue of $e^{-i\theta}u$ with real value at least $1-\zeta$ (see \cref{figu:Eigen}). This in turn implies that $u$ possesses an eigenvalue $e^{i\phi}$ such that
	\[ \Re e^{i(\phi-\theta)}=\cos(\phi-\theta)\geq 1-\zeta. \]
\end{proof}

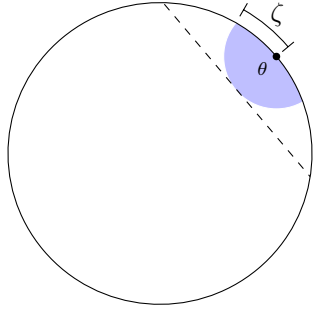
\begin{figure}[t!]
	\centering
	\ifcompile
	\tikzsetnextfilename{Eigen}
	\def\a{20}
	\begin{tikzpicture}[scale=2,rotate=40]
	\begin{scope}
		\clip (0,0) circle (1);
		\coordinate (a) at (\a:1);
		\node[draw=white,fill=blue!25] at (0:1) [circle through={(a)}] {};
		\draw[dashed] ({1-sin(\a)},-1) -- +(0,+2);
	\end{scope}
	\draw (0:1.05) -- (0:1.15) (\a:1.05) -- (\a:1.15) (0:1.1) arc(0:\a:1.1);
	\node at (\a/2:1.2) {$\zeta$};
	\node at (0:.88) {$\scriptstyle\theta$};
	\draw circle (1);
	\fill (0:1) circle (.025);
	\end{tikzpicture}
	\else
	\includegraphics{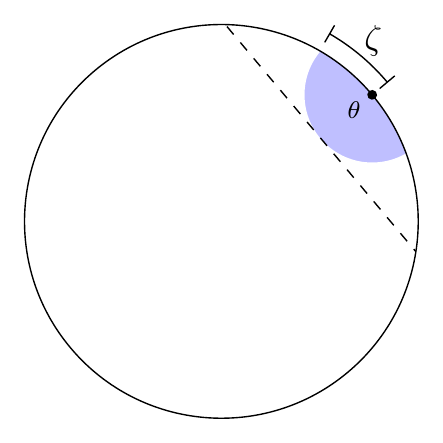}
	\fi
	\caption{\cref{lem:t-col:eigen1} gives that if there exists an expectation value in the blue region, there must exist an eigenvalue within the minor segment indicated by the dotted line.}
	\label{figu:Eigen}
\end{figure}

Combining the two above lemmas, we can \add{place a }lower bound\add{ on} the number of distinct eigenvalues of $u$.

\begin{thm2}
	If $u$ and $v$ are unitaries such that
	\[ \norm{\comm{u}{v}_{1/d}}<\frac{2}{d-1}\Bigl[1-\cos\pi/d\Bigr],  \]
	then the dimension of each operator is at least $d$.
\end{thm2}
\begin{proof}
	From \cref{lem:t-col:change1} we know that $\abs{\braopket{j}{u}{j}-e^{2i\pi j/d}}\leq \abs{j}\delta$. Applying \cref{lem:t-col:eigen1} we therefore get that $u$ must have a corresponding eigenvalue $e^{i\phi_j}$ where
	\[ \abs{\phi_j-2j\pi/d}\leq \cos^{-1}\left(1-\abs{j}\delta\right). \]
	As such we can see that each eigenvalue is within some error of a $d$th root of unity. 
	
	Next we want to find a bound for $\delta$ which ensures that these eigenvalues must be distinct, by bounding the regions in which these eigenvalues must exist away from each other. To do this we need $\abs{\phi_j-\phi_k}>0$ for all $j\neq k$. 
Taking the worst case over $j\neq k$:
	\begin{align*}
		\abs{\phi_j-\phi_k}
		&=\abs{\frac{2\pi}{d}(j-k)+\left(\phi_j-\frac{2j\pi}{d}\right)-\left(\phi_k-\frac{2k\pi}{d}\right)}\\
		&\geq\frac{2\pi}{d}\abs{j-k}-\abs{\phi_j-\frac{2j\pi}{d}}-\abs{\phi_k-\frac{2k\pi}{d}}\\
		&\geq\frac{2\pi}{d}-\cos^{-1}\left(1-\left\lceil\frac{d-1}{2}\right\rceil\delta\right)-\cos^{-1}\left(1-\left\lfloor\frac{d-1}{2}\right\rfloor\delta\right)\,.
	\end{align*} 
Here the last line follows from the fact that $j$ and $k$ cannot both saturate the worst-case distance of $\lceil \tfrac{d-1}{2} \rceil$. 
Therefore, the worst case can be chose without loss of generality to be $j = \lceil \tfrac{d-1}{2} \rceil$ and $k = - \lfloor \tfrac{d-1}{2} \rfloor$.
Using the concavity of $\cos^{-1}(z)$ over $z\in[0,1]$, we can loosen this to
\begin{align*}
	\abs{\phi_j-\phi_k}&\geq\frac{2\pi}{d}-2\cos^{-1}\left(1-\frac{d-1}{2}\,\delta\right)\,.
\end{align*}
Clearly this step is trivial for odd $d$.
	
Thus we get that a sufficient condition for all of the eigenvalues to be distinct is that the right-hand side of this inequality is strictly positive, and therefore we have the equivalent condition
\[ \cos^{-1}\left(1-\frac{d-1}{2}\,\delta\right)< \frac{\pi}{d}\,. \]
Rearranging, we find the specified bound on $\delta$ of
\[ \delta<\frac{2}{d-1}\Bigl[1-\cos(\pi/d)\Bigr]. \] 
	
%
%
%
%
%
\end{proof}

Above we have only considered the case $d=1/\alpha$, similar analysis could be performed for bounds required to certify dimensions $d' \not= 1/\alpha$. In \cref{app:lowerbound} we describe an algorithm for calculating which dimensions can be certified for an arbitrary pair of parameters $\alpha$ and $\delta$ --- running this algorithm gives \cref{fig:mountains}.

\begin{figure}
	\centering
	\ifcompile
	\tikzexternalenable
	\setlength{\figheight}{0.5\textwidth}
	\setlength{\figwidth}{0.8\textwidth}
	\tikzsetnextfilename{Mountains2}
	\input{Mountains2.tex}
	\else
	\includegraphics{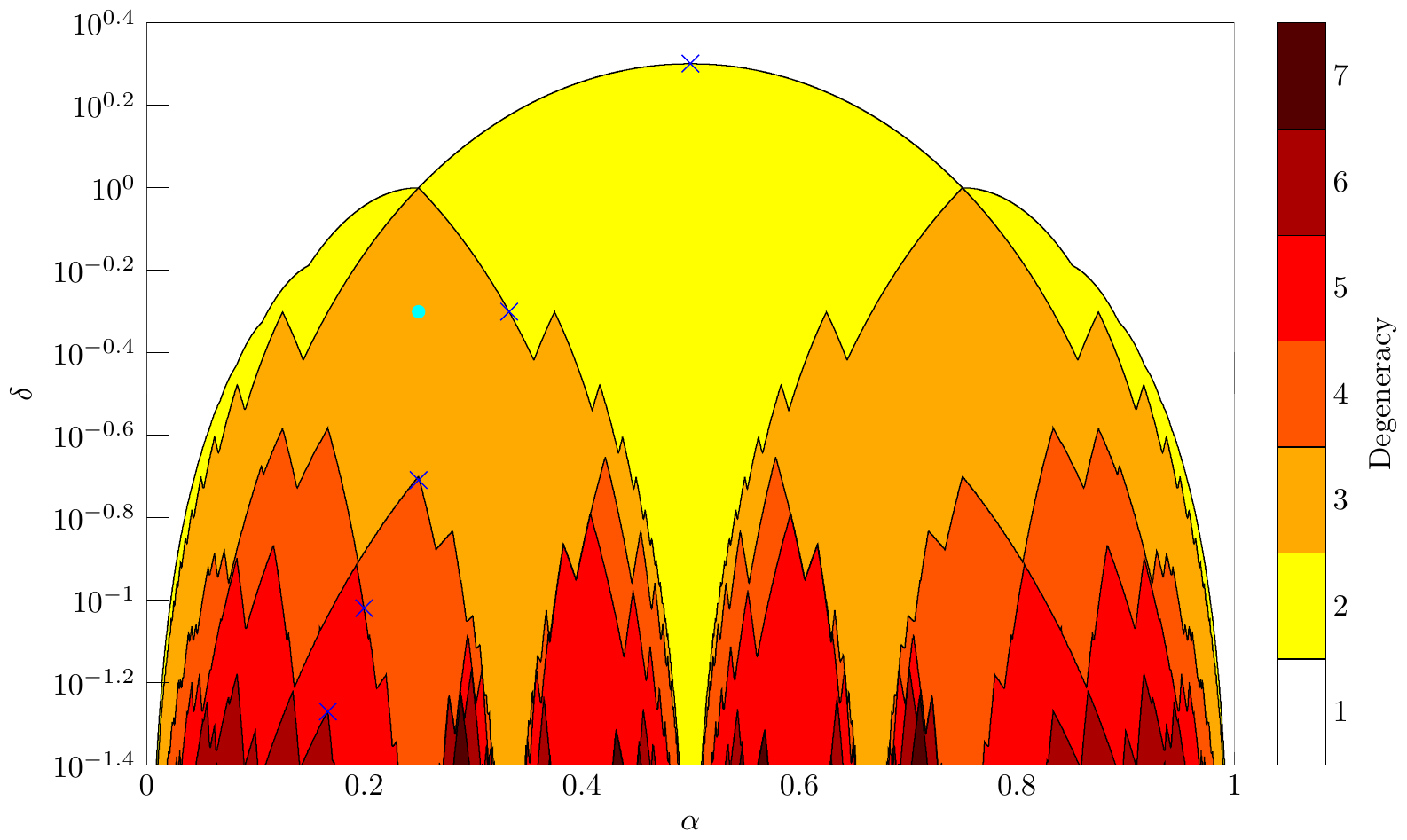}
	\fi
	\caption{The dimension that can be certified, as a function of the twisted 
	commutator value and twisting parameter\add{, i.e.\ the minimum possible 
	dimension of operators $u$ and $v$ for which $\norm{\twicomm{u}{v}}\leq 
	\delta$ as a function of $\alpha$ and $\delta$}. The blue crosses indicate 
	the \add{bounds corresponding to degeneracy $d$ and $\alpha=1/d$, 
	as}\del{level of 
	certification} considered in 
	\add{\cref{thm:single}}\del{\cref{subsec:single}}. The algorithm for 
	calculating this figure is demonstrated in \cref{app:lowerbound}, by 
	considering how the certification is calculated at the turquoise dot 
	($\alpha=1/4$ and $\delta=1/2$).}
	\label{fig:mountains}
\end{figure} 

\subsection{Two twisted pairs}
\label{subsec:double}

Next we are going to argue that the above analysis can be extended to more general collections of twisted commuting symmetries. By way of example, we are going to consider the case of two twisted pairs
\[  \norm{\comm{u_1}{v_1}_{1/d_1}}\leq \delta,\qquad\norm{\comm{u_2}{v_2}_{1/d_2}}\leq \delta, \]
each of which approximately commute
\[  \norm{\comm{u_1}{u_2}}\leq \gamma, \qquad\norm{\comm{u_1}{v_2}}\leq \delta,\qquad \norm{\comm{u_2}{v_1}}\leq \delta. \]

The equivalent of Stone-von Neumann theorem laid our in \cref{subsec:exact} gives that for $\gamma=\delta=0$, the dimension of such operators must be a multiple of $d_1d_2$. We are going to give bounds on $\gamma$ and $\delta$ below which we can prove the dimension to be at least $d_1d_2$. 

Previously we bounded the dimension from below by bounding the number of distinct eigenvalues. 
This is possible because these eigenvalues imply the existence of an orthonormal set of associated eigenvectors. As $u_1$ and $u_2$ do not commute, they will not necessarily possess an orthonormal set of shared eigenvectors. Instead we will have to address these vectors more directly, constructing \emph{approximate shared eigenvectors} and proving their linear independence. First we will see that the approximate commutation of $u_1$ and $u_2$ can be used to demonstrate the existence of such a vector.

The existence of approximate shared eigenvectors of approximately commuting matrices was first proven in generality by Bernstein in Ref.~\cite{Bernstein1971}. Whilst Bernstein considers potentially non-normal matrices, in our case both $u_1$ and $u_2$ are unitary. In \cref{app:approxeigen} we leverage this additional structure to exponentially tighten the bounds on the approximate shared eigenvectors. One of the relevant bounds considered in \cref{app:approxeigen} gives the following immediate corollary.
\begin{lem}[Approximate eigenvector]\label{lem:t-col:startvec2}
	There exists a vector $\ket{\psi}$ such that, after multiplying $u_1$ and $u_2$ by appropriate phase factor, it is an approximate shared $+1$-eigenvector of both, namely that
	\[ \norm{u_1\ket{\psi}-\ket{\psi}},\norm{u_2\ket{\psi}-\ket{\psi}}\leq \sqrt{\gamma}d_1d_2/2. \]
\end{lem}
\begin{proof}
	Given an assumption that the dimension is at most $d_1d_2$, this is a direct application of \cref{appthm:fixed}, which we consider in detail in \cref{app:approxeigen}.	
	%
	%
\end{proof}

As in the case of a single pair, we will then consider the orbit of this vector under the action of products of $v_1$ and $v_2$. Let $\ket{i,j}:=v_1^iv_2^j\ket{\psi}$ for $i=-\left\lfloor\frac{d_1-1}{2}\right\rfloor,\ldots,\left\lceil\frac{d_1-1}{2}\right\rceil$ and $j=-\left\lfloor\frac{d_2-1}{2}\right\rfloor,\ldots,\left\lceil\frac{d_2-1}{2}\right\rceil$. For convenience once again let $\eta_i:=e^{2i\pi/d_i}$.

\begin{lem}[Change in expectation value: two pair]\label{lem:t-col:change2}
	The states $\ket{i,j}$ are shared approximate eigenstates of $u_1$ and $u_2$. Specifically their approximate eigenvalues are the corresponding powers of $\eta_1$ and $\eta_2$
	\[ \abs{\phantom{\Big.}\!\braopket{i,j}{u_1}{i,j}-\eta_1^i},\abs{\braopket{i,j}{u_2}{i,j}-\eta_2^j}\leq \sqrt{\gamma}d_1d_2/2+\left(\abs{i}+\abs{j}\right)\delta. \]
\end{lem}
\begin{proof}
	From \cref{lem:t-col:startvec2} we have
	\[ \abs{\braopket{\psi}{u_1}{\psi}-1}\leq\sqrt{\gamma}d_1d_2/2. \]
	Applying an argument similar to that in \cref{lem:t-col:change1} we can bound the change in eigenvalue under the action of $v_c$ as
	\[ \abs{\braopket{i,0}{u_1}{i,0}-\eta_1 ^i\braopket{\psi}{u_1}{\psi}}\leq \abs{i}\delta. \]
	Applying the same argument for $v_2$ gives
	\[ \abs{\braopket{i,j}{u_1}{i,j}-\braopket{i,0}{u_1}{i,0}}\leq \abs{j}\delta. \]
	The triangle inequality allows us to merge these three inequalities, giving the stated bound. A similar argument can be performed for $u_2$.
\end{proof}

In the single pair case, we used the expectation values to imply the existence of nearby eigenvalues. Due to the lack of a shared eigenbasis of $u_1$ and $u_2$, we cannot do the same in the two pair case. 

The reason that a set of distinct eigenvalues lower bounds the dimension is that, for normal operators such as unitaries, the eigenvalues imply the existence of an orthonormal eigenbasis. Instead of proving the existence of such vectors indirectly through the eigenvalues, we could instead prove our vectors $\lbrace\ket{i,j}\rbrace$ to be linearly independent --- this is the approach we will take.

To this end, we will start by showing two approximate eigenvectors of a unitary with inconsistent expectation values are approximately orthogonal.

\begin{lem}[Low overlap]\label{lem:t-col:lowoverlap2}
	If two normalized vectors $\ket{x}$ and $\ket{y}$ have expectation values with some unitary $w$ such that
	\[\abs{\braopket{x}{w}{x}-e^{i\theta_x}}\leq\zeta \qquad \text{and}\qquad\abs{\braopket{y}{w}{y}-e^{i\theta_y}}\leq\zeta \]
	then the two vectors have a bounded overlap
	\[\abs{\braket{x}{y}}\leq\sqrt{2\zeta}\abs{\csc\left(\frac{\theta_y-\theta_x}{4}\right)}. \]
\end{lem}
\begin{proof}
	Firstly, let $w':=e^{-i\theta_x}w$ and $\theta:=\theta_y-\theta_x$. Next consider splitting the unit circle into three arcs $X$, $Y$, and $Z$. We let $X$ and $Y$ be centered on $\theta_x$ and $\theta_y$ respectively, and define them to be the largest possible regions such that they remain disjoint. We define $Z$ to be the remaining arc, as shown in \cref{fig:Overlap}. Note that by convexity any linear combination of eigenvectors whose eigenvalues lie in $X$ will have an expectation value in the segment subtended by $X$, and similar for $Y$.
	
	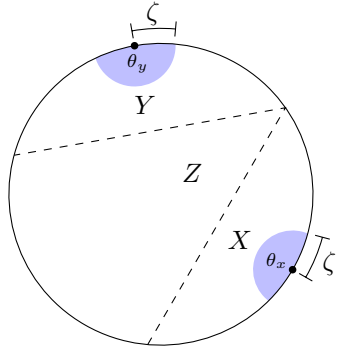
\begin{figure}[t!]
		\centering
		\ifcompile
		\tikzsetnextfilename{OverlapDiagram}
		\def\t{130}
		\def\a{15}
		\begin{tikzpicture}[scale=2,rotate=-30]
		\draw[dashed] (-\t/2:1) -- (\t/2:1) -- (3*\t/2:1);
		\node at (180+\t/2:-.25) {$Z$};
		\node at (0:.6) {$X$};
		\node at (\t:.6) {$Y$};
		\begin{scope}
		\clip (0,0) circle (1);
		\fill[fill=blue!25] (0:1) circle (0.26);
		\fill[fill=blue!25] (\t:1) circle (0.27);
		\end{scope}
		\draw (0:1.05) -- (0:1.15) (\a:1.05) -- (\a:1.15) (0:1.1) arc(0:\a:1.1);
		\draw[rotate=\t-\a] (0:1.05) -- (0:1.15) (\a:1.05) -- (\a:1.15) (0:1.1) arc(0:\a:1.1);
		\node at (\a/2:1.2) {$\zeta$};
		\node at (\t-\a/2:1.2) {$\zeta$};
		\node at (0:.88) {$\scriptstyle\theta_x$};
		\node at (\t:.88) {$\scriptstyle\theta_y$};
		\draw circle (1);
		\fill (0:1) circle (.025) (\t:1) circle (.025);
		\end{tikzpicture}
		\else
		\includegraphics{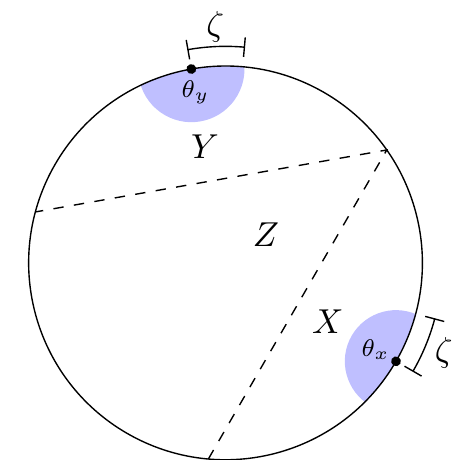}
		\fi
		\caption{A disc representing the expectation values of vectors with respect to $w$, as well as the three regions $X$, $Y$, $Z$ into which the disc is divided. The expectation values with respect to $\ket{x}$ and $\ket{y}$ lie in each of the blue regions.}
		\label{fig:Overlap}
	\end{figure}
	
	Now split $\ket{x}$ into two components 
	\[ \ket{x}=\sqrt{1-\lambda_x}\ket{x_{X}}+\sqrt{\lambda_x}\ket{x_{YZ}}, \]
	where $\ket{x_X}$ is in the span of eigenvectors with values in $X$, and $\ket{x_{YZ}}$ similar for $Y\cup Z$. By definition of $X$, we have
	\[ \Re \braopket{x_{YZ}}{w'}{x_{YZ}}\leq \cos(\theta/2)\leq \Re \braopket{x_X}{w'}{x_X}\leq 1. \]
	Next we use the bound on the expectation value.
	\begin{align*}
		\zeta
		&\geq \abs{\braopket{x}{w'}{x}-1}\\
		&\geq 1-\Re\braopket{x}{w'}{x}\\
		&= 1-(1-\lambda_x)\Re\braopket{x_X}{w'}{x_X}-\lambda_x\Re\braopket{x_{YZ}}{w'}{x_{YZ}}\\
		&\geq 1-(1-\lambda_x)-\lambda_x\cos\left(\theta/2\right)\\
		&=2\lambda_x\sin^2(\theta/4)
	\end{align*}
	Thus we conclude that $\lambda_x \leq (\zeta/2)\csc^2(\theta/4)$. Similarly if we were to have decomposed $\ket{y}$ into parts contained in $Y$ and $XZ$ as
	$\ket{y}=\sqrt{1-\lambda_y}\ket{y_{Y}}+\sqrt{\lambda_y}\ket{y_{XZ}}$
	then $\lambda_y \leq (\zeta/2)\csc^2(\theta/4)$. 
	
	Further decomposing 
	\[ \ket{x_{YZ}}=\cos\varphi_x\ket{x_Y}+\sin\varphi_x\ket{x_Z}\qquad \ket{y_{XZ}}=\cos\varphi_y\ket{y_X}+\sin\varphi_y\ket{y_Z}\,, \]
	then the inner product has the form
	\begin{align*}
		\abs{\braket{x}{y}}
		&=\Bigl|\sqrt{1-\lambda_x}\sqrt{\lambda_y}\cos\varphi_y\braket{x_X}{y_{X}}
		+\sqrt{\lambda_x}\sqrt{1-\lambda_y}\cos\varphi_x\braket{x_{Y}}{y_Y}
		+\sqrt{\lambda_x}\sqrt{\lambda_y}\sin\varphi_x\sin\varphi_y\braket{x_{Z}}{y_{Z}}\Bigr|\\
		&\leq \sqrt{1-\lambda_x}\sqrt{\lambda_y}\cos\varphi_y
		+\sqrt{\lambda_x}\sqrt{1-\lambda_y}\cos\varphi_x
		+\sqrt{\lambda_x}\sqrt{\lambda_y}\sin\varphi_x\sin\varphi_y.
	\end{align*}
	Using the identity $\abs{A\cos\phi+B\sin \phi}^2\leq \abs{A}^2+\abs{B}^2$, we can maximize over $\varphi_x$ to get
	\begin{align*}
		\abs{\braket{x}{y}}&
		\leq \sqrt{1-\lambda_x}\sqrt{\lambda_y}\cos\varphi_y
		+\sqrt{\lambda_x}\sqrt{1-\lambda_y\cos^2\varphi_y}.
	\end{align*}	
	Using $\cos\varphi_y\leq 1$, we can simplify this bound to
	\begin{align*}
	\abs{\braket{x}{y}}&
	\leq \sqrt{\lambda_y}
	+\sqrt{\lambda_x}.
	\end{align*}
	Applying the $\zeta$-dependent bounds on the $\lambda$ values, we get the stated bounds.
\end{proof}

Now that we have a way of bounding the overlap between our vectors, we need to determine how low this overlap needs to be before linear independence can be ensured.

\begin{lem}[Overlap threshold]\label{lem:t-col:threshold2}
	Take a set of normalized vectors $S=\lbrace \ket{v_i}\rbrace$ for $1\leq i \leq n$. If the pairwise overlap between any two vectors is bounded $\abs{\braket{v_i}{v_j}}<1/(n-1)$ for $i\neq j$, then $S$ is linearly independent. 
\end{lem}
\begin{proof}
	Let $G$ be the Gram matrix associated with $S$. As each of the vectors is normalized $G_{ii}=1$ for all $i$. As all of the non-diagonal entries are strictly modulus-bounded by $1/(n-1)$, this matrix is strictly diagonally dominant, i.e.\
	\[ \abs{G_{ii}}>\sum_{j\neq i}\abs{G_{ij}}\quad\text{for all } i. \]
	From the Geshgorin circle theorem, such matrices are non-singular and full rank, allowing us to conclude that $S$ is linearly independent. 
	
	Note that this analysis is tight, i.e.\ if $\braket{v_i}{v_j}=-1/(n-1)$ for all $i\neq j$ then $G$ is singular and $\sum_i \ket{v_i}=0$. By considering the eigenvectors of such a Gram matrix, a set of vectors satisfying this can be backed out.
\end{proof}

Given this bound, we can finally find the condition for our vectors to be linearly independent and therefore lower bound the dimension of the space in which they reside.

\begin{thm2}
	If $u_1$, $u_2$, $v_1$ and $v_2$ are unitaries such that they satisfy the commutation relations
	\[ \norm{\comm{u_1}{u_2}}\leq\gamma \qquad \norm{\comm{u_1}{v_2}}\leq\delta \qquad \norm{\comm{u_2}{v_1}}\leq\delta\]
	and twisted commutation relations
	\[ \norm{\comm{u_1}{v_1}_{1/d_1}}\leq\delta \qquad \norm{\comm{u_2}{v_2}_{1/d_2}}\leq\delta \]
	with $d_1\leq d_2$ and
	\[ \sqrt{\gamma}d_1d_2+(d_1+d_2)\delta< \frac{\sin^2(\pi/2d_1)}{(d_1d_2-1)^2}, \]
	then the dimension of each operator is at least $d_1d_2$.
\end{thm2}
\begin{proof}
	From \cref{lem:t-col:change2} we have that our vectors have expectation values bounded near powers of $\eta_1$ and $\eta_2$
	\[ \Bigl|\braopket{i,j}{u_1}{i,j}-\eta_1^i\Bigr|,~
	\Bigl|\braopket{i,j}{u_2}{i,j}-\eta_2^j\Bigr|
	\leq \sqrt{\gamma}d_1d_2/2+\left(\abs{i}+\abs{j}\right)\delta. \]
	
	Take a pair of vectors $\ket{i,j}$ and $\ket{i',j'}$ such that $i\neq i'$. Applying \cref{lem:t-col:lowoverlap2} with $w=u_1$ we get that their overlap is bounded as
	\begin{align*}
		\abs{\braket{i,j}{i',j'}}^2
		&\leq \Bigl[\sqrt{\gamma}d_1d_2+2\max\lbrace \abs{i}+\abs{j},\abs{i'}+\abs{j'}\rbrace\delta\Bigr]\cdot \csc^2\left(\frac{\pi(i-i')}{2d_1}\right).
	\end{align*}
	Combining this with a similar argument for $u_2$, and assuming $d_1\leq d_2$, we get that for $(i,j)\neq (i',j')$
	\begin{align*}
		\abs{\braket{i,j}{i',j'}}^2
		&\leq \Bigl[\sqrt{\gamma}d_1d_2+(d_1+d_2)\delta\Bigr]\csc^2\left(\frac{\pi}{2d_1}\right)
	\end{align*}
	Thus we can see that
	\[ \Bigl[\sqrt{\gamma}d_1d_2+(d_1+d_2)\delta\Bigr]\csc^2\left(\frac{\pi}{2d_1}\right)<\frac{1}{(d_1d_2-1)^2} \,.\]
	implies $\abs{\braket{i,j}{i',j'}}< 1/(d_1d_2-1)$ for all $(i,j)\neq (i',j')$. By \cref{lem:t-col:threshold2} this means that the collection of vectors $\lbrace \ket{i,j}\rbrace_{i,j}$ are linearly independent, constructively proving the dimensionality of the operators in question to be at least $d_1d_2$. Rearranging this gives the specified bound. 
\end{proof}


\section{Minimum twisted commutation value}
\label{sec:minimum}

In the previous section we considered finding lower bounds on the dimensions of approximately twisting commuting operators. 
In the exact case, the Stone-von Neumann theorem (c.f.\ \cref{thm:svnt}) tell us that unitaries $x$ and $y$ for which
\[\comm{x}{y}_{1/d}=0\]
are not only \emph{at least} $d$-dimensional, but are \emph{a multiple of} $d$-dimensional. 
We might therefore hope for a more comprehensive understanding of twisted commutation that provides more information than simply a lower bound on the dimension. 
In this section we will consider the twisted commutator in the Schatten\add{-Ky 
Fan} norms $\normu{\cdot}:=\norm{\cdot}_{(p,k)}$ with $p\geq 2$, and find the 
minimum possible twisted commutator value as a function of dimension.

\begin{defn}[Minimum twisted commutator value]
	\label{def:lambda}
	Let $\Lambda^{(p,k)}_{g,\alpha}$ be the \emph{minimum twisted commutator value}, with respect to the Schatten\add{-Ky Fan} $\add (p\add{,k)}$-norm, over all \add{pairs of }unitary matrices of dimension $g$ 
	\[ \Lambda_{g,\alpha}^{(p,k)}:=\min_{u,v\in{U}(g)}\normpk{\twicomm{u}{v}}. \]
\end{defn}

In this language, the Stone-von Neumann theorem gives that $\Lambda^{(p\add{,k})}_{g,\alpha}=0$ if and only if $g\alpha\in\mathbb{Z}$. If we had an understanding of the values of $\Lambda^{(p\add{,k})}_{g,\alpha}$ where $g\alpha\notin\mathbb{Z}$, then we could use twisted commutation value as a way of certifying dimension. 
In particular, if one thinks of $\alpha$ as fixed, and one knows the value $\normpk{\twicomm{u}{v}}$ to be less than $\Lambda^{(p\add{,k})}_{g,\alpha}$ for certain dimensions $g$, then these certain dimensions are ruled out as possible dimensions of $u$ and $v$. 
In this section we will explicitly evaluate $\Lambda^{(p\add{,k})}_{g,\alpha}$\add{ for $p\geq 2$}.

%

To lower bound $\Lambda_{g,\alpha}^{(p\add{,k})}$, we will utilize techniques from spectral perturbation theory to bound a related quantity known as the \emph{spectral distance}. 
By considering a family of operators which twisted commute, we will furthermore show this bound to be tight.

\begin{defn}[Spectral distance]
	The \emph{spectral $\add(p\add{,k)}$-distance} $d_{\add(p\add{,k)}}(a,b)$ 
	between two matrices $a$ and $b$ is the $\add(p\add{,k)}$-norm of the 
	vector containing 
	the differences between eigenvalues of the two matrices, minimized over all 
	possible orderings. If we let $\lambda(x)$ denote the vector of eigenvalues 
	of a $g \times g$ matrix $x$ then algebraically
	\[ d_{(p,k)}(a,b):=
	\min_{\sigma\in S_g} \normpk{\sigma\left[\lambda(a)\right]-\lambda(b)}=\min_{\sigma\in S_g}\left(\sum_{j=1}^{k}\abs{\lambda_{\sigma(j)}(a)-\lambda_j(b)}^p\right)^{1/p}\,,\]
where the minimization is over all elements $\sigma$ of the permutation group $S_g$ on $g$ symbols.
\end{defn}

\subsection{Frobenius spectral bound}

Before attacking the spectral distance, we are first going to restrict ourselves to the case of the Frobenius norm ($p=2$, $k=g$), where we shall denote the norm by $\normF{\cdot}$, the corresponding spectral distance by $d_{F}(\cdot,\cdot)$, and the twisted commutator minimum by $\Lambda_{g,\alpha}^{(F)}$. 
In this special case, the spectral distance between two normal matrices is bounded by their norm difference. 

\begin{lem}[Wielandt-Hoffman inequality~\cite{HoffmanWielandt1953}]
	For normal matrices $a$ and $b$, $d_{F}(a,b)\leq \normF{a-b}$.
	\label{lem:wh}
\end{lem}

Once again let $\eta:=e^{2i\pi \alpha}$. Applying Wielandt-Hoffman to $\Lambda_{g,\alpha}^{(F)}$ we see that the corresponding spectral distance provides a lower bound,
\[ \Lambda_{g,\alpha}^{(F)}=\min_{u,v\in U(d)}\normF{v^\dag u v-\eta u}\geq\min_{u,v\in U(g)}d_{F}(v^\dagger uv,\eta u)=\min_{u\in U(g)}d_{F}(u,\eta u) \,.\]
Though $\normF{v^\dag u v-\eta u}$ depended on both $u$ and $v$, $d_F(u,\eta u)$ depends only on the spectrum of $u$, making for a much simpler optimization. 
This inequality will turn out to be tight for matrices minimizing the twisted commutator value.

Denote the eigenvalues of $u$ by $\lbrace e^{i\theta_j}\rbrace$, then the spectral distance in question is given by
\[
d_{F}^2(u,\eta u)
:=\min_{\sigma\in S_g}\sum_{j=1}^g \abs{e^{i\theta_{\sigma(j)}}-e^{i(\theta_{j}+2\pi \alpha)}}^2
=\min_{\sigma\in S_g}\sum_{j=1}^g 4\sin^2\left(\frac{\theta_{\sigma(j)}-\theta_j-2\pi\alpha}{2}\right).
\]
Define $f(\sigma;\theta_1,\dots,\theta_g)$ to be the argument of the above optimization
\begin{align}
	f(\sigma;\theta_1,\dots,\theta_g):=\sum_{j=1}^g 4\sin^2\left(\frac{\theta_{\sigma(j)}-\theta_j-2\pi\alpha}{2}\right) \label{eqn:fdef}
\end{align}
such that $d_F^2(u,\eta u)=\min_{\sigma}f(\sigma;\theta_1,\dots,\theta_g)$ . 
The optimization of $d_F^2(u,\eta u)$ can therefore be reduced to an optimization of $f(\sigma;\theta_1,\dots,\theta_g)$.

We can now break the optimization of $f$ down into two parts. 
First we will show that for any assignment of permutation and angles, there exists a cyclic permutation, and adjusted angles, for which the value of $f$ is the same. 
This will allow us to consider a minimizing permutation which has only a single cycle without loss of generality. 
Secondly we shall see that, for such a cyclic permutation, the set of angles which minimize $f$ are those that are equally distributed around the unit circle. 
Given these, we will find an explicit minimum for $f$, and thus for $d_F(u,\eta u)$.

\begin{lem}[Reduction to cyclic permutations]
	\label{lem:optimal_permutation}
	For a given multi-cycle permutation $\sigma$ and set of angles $\lbrace \theta_j\rbrace$, there exists a cyclic permutation $\sigma'$ and set of adjusted angles $\lbrace \theta_j'\rbrace$ such that 
	\[ f(\sigma;\theta_1,\dots,\theta_g)=f(\sigma';\theta_1',\dots,\theta_g'). \]
\end{lem}
\begin{proof}
	Firstly, our indices can be reordered such that the cycles of $\sigma$ are contiguous, i.e.\ in cycle notation
	\[\sigma=(1~\ldots~k_1-1)\,(k_1~\ldots~k_2-1)\ldots(k_n~\ldots~g),\]
	for some $1<k_1\dots<k_n\leq g$. 
	(Note that the result is trivially true if $g=1$, so we restrict to $g>1$.) 
	As $f$ only depends on the difference between angles whose indices are within the same cycle of $\sigma$, if we shift all the angles within the same cycle by the same amount, the value of $f$ will not change. For example  if we take the change of angle
	\[ \theta_j':=\begin{dcases}
	\theta_j-\theta_1 & \,\,\,1\leq j< k_1\\
	\theta_j-\theta_{k_1} & k_1\leq j< k_2\\
	~\quad\vdots & \\
	\theta_j-\theta_{k_n} & k_n\leq j\leq g. 
	\end{dcases} \]
	then $f(\sigma;\theta_1,\dots,\theta_g)=f(\sigma;\theta_1',\dots,\theta_g')$. Notice that $\theta_1'=\theta'_{k_1}=\dots=\theta_{k_n}'=0$ by construction.
	
	We now wish to merge the permutation $\sigma$ into a single cyclic permutation
	\begin{align}
	\sigma':=(1~\ldots~g).\label{eqn:perm}
	\end{align}
	To do this, the only entries of the permutation which need to be changed are those at the end of each cycle.
	\newcommand{\arrow}{\hspace{-2cm}&\rightarrow\hspace{-2cm}}
	\begin{align*}
	\hspace{2cm}\sigma(k_1-1)&=1 \arrow& & \sigma'(k_1-1)&=k_1\\
	\sigma(k_2-1)&=k_1 \arrow& & \sigma'(k_2-1)&=k_2\\
	&~\,\vdots \hspace{-2cm}&\hspace{-2cm} & & &~\,\vdots \\
	\sigma(g)&=k_n \arrow& & \sigma'(g)&=1	\,.	
	\end{align*}
	By definition of the adjusted angles however, the only indices that change are those for which the angles have already been made identical in the previous step, i.e.\ $\theta_{\sigma(j)}'=\theta_{\sigma'(j)}'$ for all $j$. As $f$ only depends on $\sigma$ through how it acts on the angles, this means that this doesn't alter the value of $f$, therefore $f(\sigma;\theta_1',\dots,\theta_g')=f(\sigma';\theta_1',\dots,\theta_g')$.
\end{proof}

Now that we have addressed the nature of the optimal permutation, namely showing that it can be taken to be cyclic, we turn out attention to the optimal angles.

\begin{lem}
	\label{lem:optimal_angles}
	For a given single-cycle permutation $\sigma$, the sets of angles which optimize $f$, as defined in \cref{eqn:fdef}, correspond to those evenly distributed around the unit circle, and the difference between adjacent angles $\theta_j$ and $\theta_{\sigma(j)}$ is $2\pi\lfloor d\alpha\rceil/g$, where $\round{\cdot}$ denotes integer rounding. Moreover the corresponding minimal value of $f$ is 
	\[\min_{\lbrace\theta_j\rbrace_j} f(\sigma;\theta_1,\dots,\theta_g)=2\sqrt{g}\sin\left(\pi\abs{\frac{\round{g\alpha}-g\alpha}{g}}\right). \]
\end{lem}
\begin{proof}
	Denote both of the terms\footnote{In saying there are two such terms we have assumed $g\geq 3$. If $g=1$ the lemma is trivial ($f$ is constant), and if $g=2$ then we have double counted in $f_j(\theta_j)$, but our analysis of its minimum remains valid.} in $f$ which depend non-trivially on $\theta_j$ by $f_j(\theta_j)$. Using the double angle formula and the auxiliary angle method, we can reduce the $\theta_j$ dependence to a single sinusoidal term.
	\begin{align*}
	f_j(\theta_j)&=4\sin^2\left(\frac{\theta_j-\theta_{\sigma(j)}-2\pi\alpha}{2}\right)+4\sin^2\left(\frac{\theta_{\sigma^{-1}(j)}-\theta_{j}-2\pi\alpha}{2}\right)\\
	&=4-4\cos\left(2\pi\alpha+\frac{\theta_{\sigma(j)}-\theta_{\sigma^{-1}(j)}}{2}\right) \cos\left(\theta_j-\frac{\theta_{\sigma(j)}+\theta_{\sigma^{-1}(j)}}{2}\right)\,.
	\end{align*}
	
	We can therefore see that the optimal $\theta_j$, leaving all other angles fixed, satisfies
	\[ \theta_j=\left(\theta_{\sigma(j)}+\theta_{\sigma^{-1}(j)}\right)/2~\mod{\pi}. \]
	This implies that $\theta_{\sigma(j)}-\theta_j=\theta_j-\theta_{\sigma^{-1}(j)}\mod{2\pi}$, i.e.\ $\theta_j$ lies in at the `midpoint' of its neighbors, as described by $\sigma$. By inducting the above argument we find that $\theta_{\sigma(j)}-\theta_j=\theta_{\sigma(k)}-\theta_k\mod{2\pi}$ for all $j,k$ meaning that all adjacent angles are equally spaced around the unit circle. This means that if we have $g$ angles, and label our indices such that $\sigma(j)=j+1\mod{g}$, then for some fixed integer $m$, the optimal angles are of the form
	\begin{align}
	\theta_j=\theta_1+2\pi m(j-1)/g. \label{eqn:angles}
	\end{align}
	The only free parameter left now is $m$, the spacing between adjacent points. Plugging these angles into the definition of $f$ we find
	\[ f(\sigma;\theta_1,\dots,\theta_g)=2\sqrt{g}\Bigl\vert\sin\bigl(\pi\left[m/g-\alpha\right]\bigr)\Bigr\vert. \]
	This is in turn minimized for $m=\round{g\alpha}$, giving the stated spacing and minima.
\end{proof}

As this minimum of $f$ is independent of the permutation $\sigma$, we get an overall minimum for $f$ for free.

\begin{corr}
	\label{corr:bound}
	The minimum twisted commutator value (\cref{def:lambda}) in the Frobenius norm $\Lambda_{g,\alpha}^{(F)}$ is lower bounded
	\[ \Lambda_{g,\alpha}^{(F)}\geq 2\sqrt{g}\sin\left(\pi\abs{\frac{\round{g\alpha}-g\alpha}{g}}\right). \]
\end{corr}
\begin{proof}
	This result can be seen by recalling that the definition of $f$ in \cref{eqn:fdef} gives that
	\[ \min_{u\in U(g)}d_F(u,\eta u)=\min_{\sigma,\lbrace\theta_j\rbrace_j}f(\sigma,\theta_1,\dots,\theta_g). \]
	As \cref{lem:optimal_permutation} tells us that we can consider cyclic permutations without loss of generality, we can apply the minimum found in \cref{lem:optimal_angles}, giving 
	\[\min_{u\in U(g)}d_F(u,\eta u)=2\sqrt{g}\sin\left(\pi\abs{\frac{\round{g\alpha}-g\alpha}{g}}\right).\]
	Applying the Wielandt-Hoffman theorem (\cref{lem:wh}), we get that the above minimum spectral distance lower bounds the twisted commutator in the Frobenius norm, as required.
\end{proof}

\subsection{Higher norms and tightness}

With the above bound in hand, we now turn our attention to tightness. A canonical family of operators which exhibit twisted commutation is that of the \emph{generalized Pauli operators}, also known as Sylvester's \emph{clock and shift matrices}
\[ C:=\sum_j \omega^{j-1}\ket{j}\bra{j}, \qquad\qquad S:=\sum_j\ket{j\oplus 1}\bra{j} \]
where $\omega=e^{2i\pi/g}$ is a primitive $g$th root of unity, and $\oplus$ denotes addition modulo $g$. As $S$ simply cyclically permutes the eigenbasis of $C$, we can see that $S^\dag C S=\omega C$, or $\comm{C}{S}_{1/g}=0$. By taking appropriate powers these operators can also yield pairs which twisted commute with a phase that is any power of $\omega$, specifically we see $\comm{C}{S^k}_{k/g}=0$. Suppose we take such a pair and evaluate the twisted commutator at an arbitrary phase $\eta=e^{2i\pi \alpha}$. We then find,
\begin{align*}
\normF{\twicomm{C}{S^k}}
&=\normF{CS^k-\eta S^kC}\\
&=\normF{(1-\omega^k\eta)CS^k}\\
&=\sqrt{g}\abs{1-\omega^k\eta}\\
&=2\sqrt{g}\Bigl\vert\sin\bigl(\pi\left(\alpha+k/g\right)\bigr)\Bigr\vert.
\end{align*}
If we now take $k=-\round{g\alpha}$, then we saturate \cref{corr:bound}, proving tightness of the bound on $\Lambda_{g,\alpha}^{(F)}$, allowing us to conclude
\begin{align*}
\Lambda_{g,\alpha}^{(F)}=2\sqrt{g}\sin\left(\pi\abs{\frac{\round{g\alpha}-g\alpha}{g}}\right).
\end{align*}

For the above optimizations we restricted ourself to the $p=2$ case of the Frobenius norm. The nature of the minimizers found allows us to pull this analysis up into minima for the $p>2$ Schatten norms as well. 

\begin{thm2}[Minimum twisted commutation value]
	Suppose that $u$ and $v$ are $g$-dimensional unitaries, \add{then for any $p\geq2$ the twisted commutator is lower bounded 
		\[ \Bigl\|\twicomm{u}{v}\Bigr\|_{(p,k)}\geq 2k^{1/p}\sin\left(\pi \abs{\frac{\round{g\alpha}-g\alpha}{g}}\right), \]
		where $\normpk{\cdot}$ is the $(p,k)$-Schatten-Ky Fan norm.} 
	\del{then for any $p\geq 2$ we have
		\[ \Bigl\|\twicomm{u}{v}\Bigr\|_{p}\geq 2g^{1/p}\sin\left(\pi \abs{\frac{\round{g\alpha}-g\alpha}{g}}\right). \]}
	Moreover this bound is tight, in that sense that there exist \add{families of }$g$-dimensional unitaries which saturate the above bound\add{s} and only depend on $\round{g\alpha}$, the nearest integer to $g\alpha$.
\end{thm2}
\begin{proof}
	By the equivalence of Schatten\add{-Ky Fan} norms, the minimum Frobenius norm will also \add{provide}\del{imply} a lower bound for other $\add(p\add{,k)}$-norms as well. Specifically for $p\geq2$ we have
	\begin{align*}
	\normpk{M}\geq \add{k^{1/p}}g^{\del{1/p}-1/2}\normF{M}. 
	\end{align*}
	
	Applying these to the definition of $\Lambda_{g,\alpha}^{(p,k)}$, this 
	bound gives that $\Lambda_{g,\alpha}^{(p.k)}\geq 
	\add{k^{1/p}}g^{\del{1/p}-1/2}\Lambda_{g,\alpha}^{(F)}$ for $p\geq 2$. It 
	turns out that this inequality is saturated by matrices $M$ with flat 
	spectra, i.e.\ those proportional to unitaries. It so happens that the 
	clock and shift operators considered to demonstrate tightness have a 
	twisted commutator with precisely this property, and therefore also 
	saturate and demonstrate the tightness of the induced $p\add>\del\geq 2$ 
	bound\add{s}. We therefore conclude that
	\[ \Lambda_{g,\alpha}^{(p,k)}=\add{k^{1/p}}g^{\del{1/p}-1/2}\Lambda_{g,\alpha}^{(F)}=2\add{k^{1/p}}\del{g^{1/p}}\sin\left(\pi\abs{\frac{\round{g\alpha}-g\alpha}{g}}\right). \]
\end{proof}
Some plots of this bound are shown in \cref{fig:minima}.


\begin{figure}[t!]
	\centering
	a) \hspace{.475\textwidth} b) \hspace{.425\textwidth}~\vspace{-.5cm} \\
	\ifcompile
	\setlength{\figheight}{0.325\textwidth}
	\setlength{\figwidth}{0.4\textwidth}
	\tikzexternalenable
	\tikzsetnextfilename{Minima1}
	\input{./Minima1.tex}
	~
	\tikzsetnextfilename{Minima2}
	\input{./Minima2.tex}
	\else
	\includegraphics{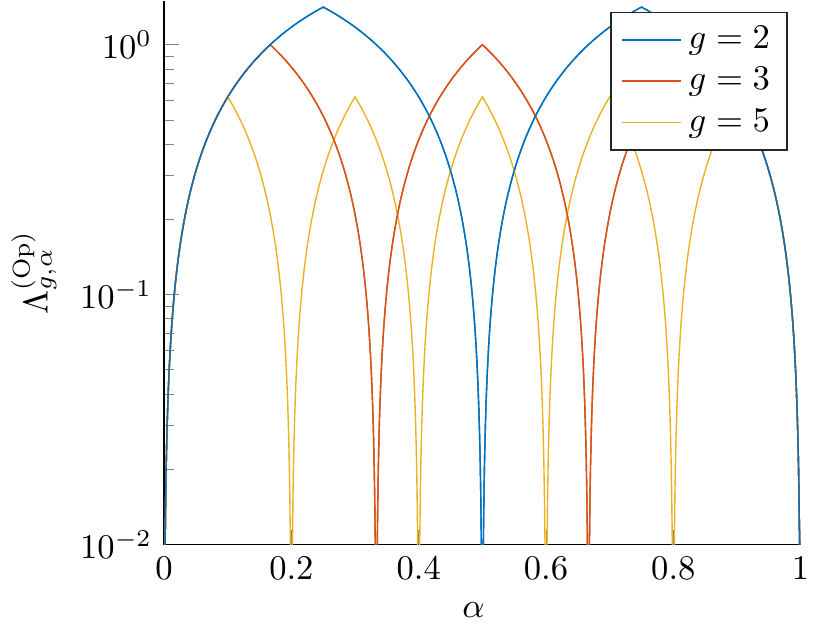}
	\quad
	\includegraphics{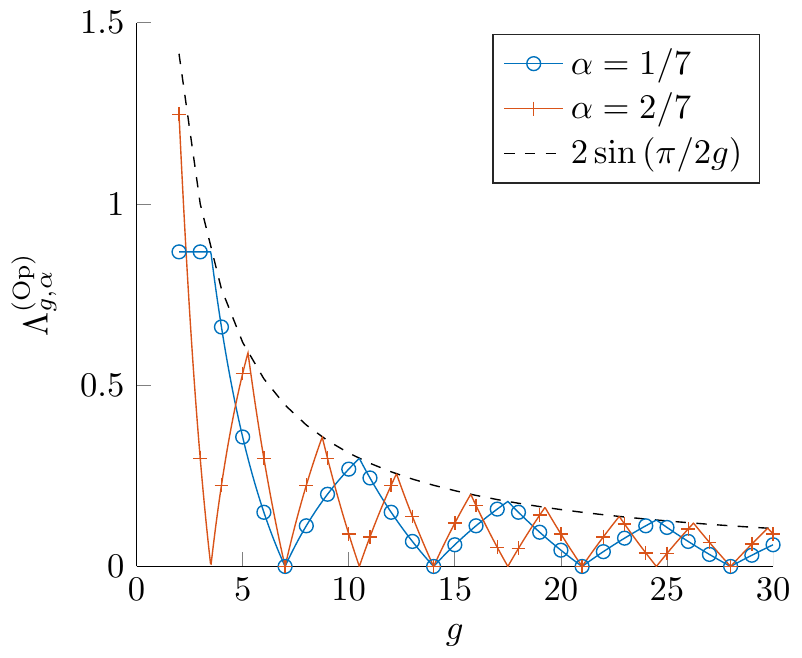}
	\fi	
	\caption{The twisted commutator value minimum (in the operator norm, $p=\infty$) $\Lambda_{g,\alpha}^{(\mathrm{Op})}$. a) The dependence on the twisting parameter $\alpha$ for a few fixed dimensions $g$. The presence of roots at multiples of $1/g$ are those predicted by \cref{thm:svnt}. b) Now fixing the twisting parameter $\alpha$, the dependence on the dimension $g$ is shown. Note that $g$ can only take integer values, indicated by the circles and pluses, with the continuous lines simply intended to guide the eye. The dotted black line indicates an $\alpha$-independent upper bound on $\Lambda_{g,\alpha}^{(\mathrm{Op})}$ given by applying the bound $\abs{x-\round{x}}\leq 1/2$.}
	\label{fig:minima}
\end{figure}


\section{Applications and open questions}
\label{sec:localapplications}

We now discuss several avenues for improvements, generalizations, refinements, and applications of these ideas. 

\subsection{Local Hamiltonians}

In this paper the only assumption we made about our Hamiltonian $H$ was the presence of a spectral gap. A natural additional structure to impose is that $H$ be a many-body Hamiltonian: decompose our Hilbert space into a tensor product of many smaller Hilbert spaces, and let our Hamiltonian take the form
\[ H=\sum_{k} h_k \]
where each term $h_k$ acts non-trivially on a constant number of these tensor factor spaces. Additional to this we could also impose that the factors on which it acts are geometrically local as well. Under this special case it may be that either the bounds on degeneracy certification might be able to be improved, or we might be able to prove the existence of degeneracy witnesses with additional structure, e.g.\ such witnesses might act in a geometrically local fashion.

\subsection{Topologically ordered systems}

While the notions of approximate symmetry and degeneracy of a ground band are both robust to small perturbations, na\"{i}vely one can only consider perturbations of a strength no larger than the gap. For topologically ordered systems~\cite{Wen2012} however, we can afford much larger perturbations under certain locality assumptions.

Under the influence of local perturbations, the low-energy band structure, most notably the ground space degeneracy, is robust even if the overall strength of the perturbation is extensive~\cite{Bravyi2010a}. Moreover, any symmetries which witnesses this degeneracy can be quasi-adiabatically continued~\cite{Hastings2005} into approximate symmetries which witness the degeneracy of the ground band in the perturbed system. It is in this sense that the existence of degeneracy witnesses can be considered robust to even rather strong perturbations, at least for the ground band.

The family of \emph{abelian quantum double models} possess symmetries supported on quasi-1D regions which satisfy twisted commutation relations related to the braid and fusion rules of the underlying anyons~\cite{Kitaev2003}. More general models such as non-abelian/twisted quantum doubles~\cite{Kitaev2003,Hu2012,Dijkgraaf1990}, and Levin-Wen string net models~\cite{LevinWen2005} are all believed to possess symmetries which satisfy more general commutation-like relations based on more general notions of commutation. One possible example is the twist product~\cite{Haah2014} which only commutes the two operators on part of the system, braiding them together.
\[ \Bigl(\sum_{i}A_i\otimes A_i'\Bigr)\infty\Bigl(\sum_{j}B_j\otimes B_j'\Bigr):=
\sum_{ij}A_iB_j\otimes B_j'A_i'.  \]
An obvious extension of this work is to take various properties of these underlying systems implied by this commutation-like relations, and see if they too carry through into the regime of \emph{approximate} relations.

In a recent paper, Bridgeman et.\ al.\ sought to classify the phases of 2D topologically ordered spin systems belonging to the same phase as abelian quantum doubles~\cite{BridgemanFlammiaPoulin2016}. This was done by numerically optimizing twisted pairs of symmetries. This optimization was done over a tensor network~\cite{BridgemanChubb2016,Orus2014} ansatz of quasi-1D operators known as matrix product operators. For two operators $L$ and $R$, supported on intersecting quasi-1D regions, the cost function takes the form
\[ C(L,R;\alpha)\propto\epsilon_L^2+\epsilon_R^2+\delta^2 \]
where $\epsilon_L:=\normF{\comm{L}{H}}$, $\epsilon_R:=\normF{\comm{R}{H}}$, and $\delta=\normF{\twicomm{L}{R}}$.

Minimizing $C(L,R;\alpha)$ over $L$ and $R$ for a fixed $\alpha$, they found that in the abelian quantum doubles the minimizers were unitary, and that both $\epsilon_L$ and $\epsilon_R$ vanish to within numerical accuracy, leaving only the twisted commutator value $\delta$. By observing the values of $\alpha$ for which the minimum cost is low, they hoped to classify the topological phases of the underlying Hamiltonian. By \cref{thm:restriction} we know that, at least to within numerical accuracy, the ribbon operators found restrict down to ground symmetries with the same twisted commutation relations. In \cref{fig:numerics} we compare, for the $\mathbb{Z}_5$ quantum double model, their numerically obtained values of this twisted commutator $\delta_\textrm{min}$ with the minimal possible twisted commutator $\Lambda^{(F)}_{5,\alpha}$, showing close agreement and lending support to the efficacy of this numerical method.

\begin{figure}[t!]
	\centering
	\ifcompile
	\setlength{\figheight}{0.3\textwidth}
	\setlength{\figwidth}{0.8\textwidth}
	\tikzexternalenable
	\tikzsetnextfilename{Gothic}
	\input{./Gothic.tex}
	\else
	\includegraphics{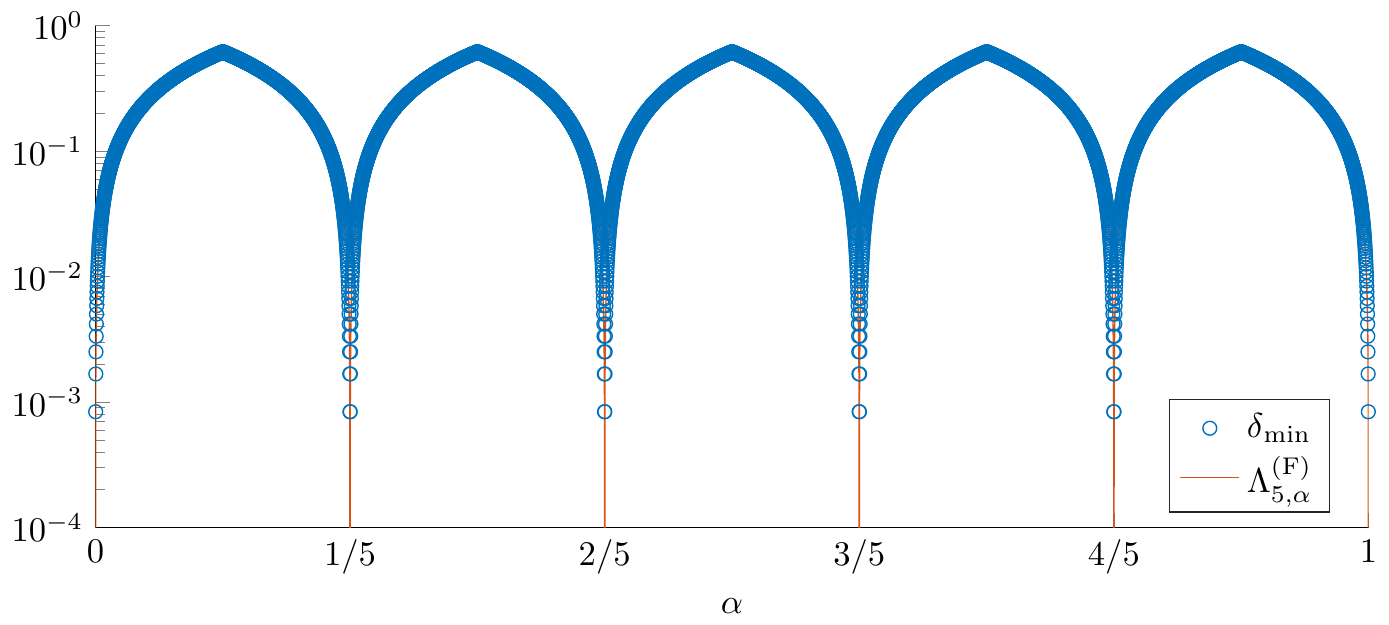}
	\fi	
	\caption{The twisted commutator value for ribbon operators on the $\mathbb{Z}_5$ quantum double model, as calculated using the algorithm of \cite{BridgemanFlammiaPoulin2016}, as compared to the minimum possible twisted commutator value in $5$-dimensions. Note that the difference between the two plots is no more than $3\times 10^{-13}$.}
	\label{fig:numerics}
\end{figure}

\subsection{Quantum codes}

One class of systems for which twisted commuting symmetries play a special role are quantum codes, in which they can be interpreted as logical operators~\cite{NielsenChuang2011,Preskill1999,Knill1996}. For a quantum code encoding $N$ codewords, the logical algebra must correspond to $\text{Mat}_N(\mathbb{C})$, which necessarily contains a pair of operators $X$ and $Z$ such that $\comm{X}{Z}_{1/N}=0$; indeed the algebra generated by any two such operators $X$ and $Z$ is itself $\text{Mat}_N(\mathbb{C})$. 

While the existence of logical operators which $\alpha=1/N$ twisted commute can be ensured, we might only see and expect operators with twisted commutations characteristic of smaller ground spaces if we restrict the locality of these operators. Though the logical algebra is given by $\text{Mat}_N(\mathbb{C})$, this space often naturally decomposes into a tensor product decomposition: the logical qudits. By geometrically restricting where on the system the operators can act, we can often restrict which factors the logical operators have nontrivial commutation relations with. This is the case for celebrated examples such as the toric code~\cite{Kitaev2003}. This can be seen above in \cref{fig:numerics}, where the logical operators are restricted to string-like regions that are only sensitive to one $\text{Mat}_5(\mathbb{C})$ factor of the larger $\text{Mat}_{25}(\mathbb{C})$ logical algebra; one of the two 5-level qudits. In the same way that Ref.~\cite{BridgemanFlammiaPoulin2016} sought to use the existence of twisted commuting symmetries to classify topological phases, how this existence varies with respect to the geometry imposed on these operators might provide a tool to probe what portion of the logical algebra is accessible on certain regions.

In the language of quantum codes, our results can be interpreted as bounds below which approximate logical operators imply the existence of a certain number of code words. A possible avenue for future work is whether there exists bounds below which not only can the number of codestates be bounded, but reliable encoding, decoding, and error correction can all be performed with these approximate logical operators. Understanding when information stored in such states is approximately preserved, as opposed to exactly preserved~\cite{Blume-Kohout2010}, could have interesting applications in approximate quantum error correction.


\subsection*{Acknowledgements}
\label{app:acknowledgements}

We thank Jacob Bridgeman for fruitful comments, and for computing the numerical 
data for \cref{fig:numerics} using the algorithm of 
Ref.~\cite{BridgemanFlammiaPoulin2016}. For the use of American English in this 
paper, CTC apologises to the Commonwealth of Australia. 
This work was supported by the Australian Research Council via EQuS project number CE11001013 and STF was supported by an Australian Research Council Future Fellowship FT130101744.

\bibliographystyle{ChubbStyle}
\small\bibliography{library}

\begin{thebibliography}{10}

\bibitem{Weyl}
H.~Weyl, ``{Das asymptotische Verteilungsgesetz der Eigenwerte linearer partieller Differentialgleichungen (mit einer Anwendung auf die Theorie der Hohlraumstrahlung)},'' {\em \href{http://dx.doi.org/10.1007/BF01456804}{Mathematische Annalen}} {\bfseries 71}, 4, 441--469, (1911).

\bibitem{Rosenthal1969}
P.~Rosenthal, ``{Are Almost Commuting Matrices Near Commuting Matrices?},''
  {\em \href{http://dx.doi.org/10.2307/2317951}{The American Mathematical
  Monthly}} {\bfseries 76}, 925,  (1969).

\bibitem{Halmos1976}
P.~R. Halmos, ``{Some unsolved problems of unknown depth about operators on
  Hilbert space},'' {\em Proceedings of the Royal Society of Edinburgh,
  Section: A Mathematics} {\bfseries 76}, 1, 67--76,  (1976).

\bibitem{Szarek1990}
S.~Szarek, ``{On almost commuting Hermitian operators},'' {\em
  \href{http://dx.doi.org/10.1216/rmjm/1181073131}{Rocky Mountain Journal of
  Mathematics}} {\bfseries 20}, 583--591,  (1990).

\bibitem{Davidson1985}
K.~R. Davidson, ``{Almost commuting hermitian matrices},'' {\em Mathematica
  Scandinavica} {\bfseries 56}, 222,  (1985).

\bibitem{Berg1991}
I.~D. Berg and K.~R. Davidson, ``{Almost commuting matrices and a quantitative
  version of the Brown-Douglas-Fillmore theorem},'' {\em
  \href{http://dx.doi.org/10.1007/BF02398885}{Acta Mathematica}} {\bfseries
  166}, 121--161,  (1991).

\bibitem{Lin1997}
H.~Lin, ``{Almost commuting selfadjoint matrices and applications},'' in {\em
  Operator algebras and their applications}, 193--223, Fields Institute
  Communications, Vol. 13,  (1997).

\bibitem{Kachkovskiy2014}
I.~Kachkovskiy and Y.~Safarov, ``{Distance to normal elements in C*-algebras of
  real rank zero},'' {\em
  \href{http://dx.doi.org/10.1090/S0894-0347-2015-00823-2}{Journal of the
  American Mathematical Society}} {\bfseries 29}, 61--80,
  \href{http://arxiv.org/abs/1403.2021}{arXiv:1403.2021},  (2015).

\bibitem{Friis1996}
P.~Friis and M.~R{\o}rdam, ``{Almost commuting self-adjoint matrices - a short
  proof of Huaxin Lin's theorem},'' {\em
  \href{http://dx.doi.org/10.1515/crll.1996.479.121}{Journal f{\"{u}}r die
  reine und angewandte Mathematik}} {\bfseries 1996}, 479, 121--131,  (1996).

\bibitem{hastings09b}
M.~B. Hastings, ``{Making Almost Commuting Matrices Commute},'' {\em
  \href{http://dx.doi.org/10.1007/s00220-009-0877-2}{Communications in
  Mathematical Physics}} {\bfseries 291}, 321--345,
  \href{http://arxiv.org/abs/0808.2474}{arXiv:0808.2474},  (2009).

\bibitem{Voiculescu1983}
D.~Voiculescu, ``{Asymptotically commuting finite rank unitary operators
  without commuting approximants},'' {\em Acta Scientiarum Mathematicarum}
  {\bfseries 45}, 429--431,  (1983).

\bibitem{Exel1989}
R.~Exel and T.~Loring, ``{Almost commuting unitary matrices},''
  \href{http://dx.doi.org/10.1090/S0002-9939-1989-0975641-9}{in {\em
  Proceedings of the American Mathematical Society}}, ~106, 913--913,  (1989).

\bibitem{Exel1991}
R.~Exel and T.~A. Loring, ``{Invariants of almost commuting unitaries},'' {\em
  \href{http://dx.doi.org/10.1016/0022-1236(91)90034-3}{Journal of Functional
  Analysis}} {\bfseries 95}, 364--376,  (1991).

\bibitem{Choi1988}
M.-D. Choi, ``{Almost Commuting Matrices Need not be Nearly Commuting},'' {\em
  \href{http://dx.doi.org/10.2307/2047216}{Proceedings of the American
  Mathematical Society}} {\bfseries 102}, 529,  (1988).

\bibitem{Lin1998}
H.~Lin, ``{Almost Commuting Unitaries and Classification of Purely Infinite
  Simple C*-Algebras},'' {\em
  \href{http://dx.doi.org/10.1006/jfan.1997.3214}{Journal of Functional
  Analysis}} {\bfseries 24}, 1--24,  (1998).

\bibitem{Osborne2008}
T.~J. Osborne, ``{Almost commuting unitaries with spectral gap are near
  commuting unitaries},''
  \href{http://arxiv.org/abs/0809.0602}{arXiv:0809.0602},  (2008).

\bibitem{Loring2013}
T.~A. Loring and A.~P.~W. S{\o}rensen, ``{Almost commuting orthogonal
  matrices},'' {\em \href{http://dx.doi.org/10.1016/j.jmaa.2014.06.019}{Journal
  of Mathematical Analysis and Applications}} {\bfseries 420}, 1051--1068,
  \href{http://arxiv.org/abs/1311.4575}{arXiv:1311.4575},  (2014).

\bibitem{Wen2012}
X.-G. Wen, ``{Topological Order: From Long-Range Entangled Quantum Matter to a
  Unified Origin of Light and Electrons},'' {\em
  \href{http://dx.doi.org/10.1155/2013/198710}{ISRN Condensed Matter Physics}}
  {\bfseries 2013}, 1--20,
  \href{http://arxiv.org/abs/1210.1281}{arXiv:1210.1281},  (2013).

\bibitem{Preskill1999}
J.~Preskill, {\em \href{http://dx.doi.org/10.1017/CBO9781139034807}{{\em
  {Quantum Error Correction}}}}.
\newblock Cambridge University Press, Cambridge,  (2013).

\bibitem{Brown2016}
B.~J. Brown, D.~Loss, J.~K. Pachos, C.~N. Self, and J.~R. Wootton, ``{Quantum
  memories at finite temperature},'' {\em
  \href{http://dx.doi.org/10.1103/RevModPhys.88.045005}{Reviews of Modern
  Physics}} {\bfseries 88}, 045005,
  \href{http://arxiv.org/abs/1411.6643}{arXiv:1411.6643},  (2016).

\bibitem{LandauLifshitz}
L.~Landau, E.~Lifshitz, V.~Berestetskii, and L.~Pitaevskii, {\em {\em {Course
  of Theoretical Physics}}}.
\newblock Pergamon Press,  (1951).

\bibitem{Brown2011}
B.~Brown, S.~T. Flammia, and N.~Schuch, ``Computational difficulty of computing
  the density of states,'' {\em
  \href{http://dx.doi.org/10.1103/physrevlett.107.040501}{Physical Review
  Letters}} {\bfseries 107},
  \href{http://arxiv.org/abs/1010.3060}{arXiv:1010.3060},  (2011).

\bibitem{ChubbFlammia2015}
C.~T. Chubb and S.~T. Flammia, ``{Computing the Degenerate Ground Space of
  Gapped Spin Chains in Polynomial Time},'' 1--32,
  \href{http://arxiv.org/abs/1502.06967}{arXiv:1502.06967},  (2015).

\bibitem{Huang2014}
Y.~Huang, ``{A polynomial-time algorithm for approximating the ground state of
  1D gapped Hamiltonians},''
  \href{http://arxiv.org/abs/1406.6355v3}{arXiv:1406.6355v3},  (2014).

\bibitem{Yang2004}
J.~Yang and H.-k. Du, ``{A note on commutativity up to a factor of bounded
  operators},'' {\em
  \href{http://dx.doi.org/10.1090/S0002-9939-04-07224-7}{Proceedings of the
  American Mathematical Society}} {\bfseries 132}, 1713--1721,  (2004).

\bibitem{Hall2013}
B.~C. Hall, {\em \href{http://dx.doi.org/10.1007/978-1-4614-7116-5}{{\em
  {Quantum Theory for Mathematicians}}}}, ~267 of {\em Graduate Texts in
  Mathematics}.
\newblock Springer New York,  (2013).

\bibitem{Rosenberg2004}
J.~Rosenberg, ``{A selective history of the Stone-von Neumann theorem},'' in
  {\em Contemporary Mathematics}, 331--353,  (2004).

\bibitem{Said2014}
M.~Said, ``{Almost Commuting Elements in Non-Commutative Symmetric Operator
  Spaces},''  (2014).

\bibitem{softtorus}
S.~Eilers and R.~Exel, ``{Finite dimensional representations of the soft
  torus},'' {\em
  \href{http://dx.doi.org/10.1090/S0002-9939-01-06150-0}{Proceedings of the
  American Mathematical Society}} {\bfseries 130}, 727--732,
  \href{http://arxiv.org/abs/math/9810165}{arXiv:math/9810165},  (1998).

\bibitem{softtorus2}
R.~Exel, ``{The soft torus and applications to almost commuting matrices},''
  {\em \href{http://dx.doi.org/10.2140/pjm.1993.160.207}{Pacific Journal of
  Mathematics}} {\bfseries 160}, 207--217,  (1993).

\bibitem{Babai}
L.~Babai and K.~Friedl, ``{Approximate representation theory of finite
  groups},'' \href{http://dx.doi.org/10.1109/SFCS.1991.185442}{in {\em
  Proceedings 32nd Annual Symposium of Foundations of Computer Science}},
  733--742, IEEE Comput. Soc. Press,  (1991).

\bibitem{Friedl}
K.~Friedl, ``{Near-representations of finite groups},''  (2003).

\bibitem{Moore2010}
C.~Moore and A.~Russell, ``{Approximate Representations, Approximate
  Homomorphisms, and Low-Dimensional Embeddings of Groups},'' {\em
  \href{http://dx.doi.org/10.1137/140958578}{SIAM Journal on Discrete
  Mathematics}} {\bfseries 29}, 182--197,
  \href{http://arxiv.org/abs/1009.6230}{arXiv:1009.6230},  (2015).

\bibitem{Bernstein1971}
A.~R. Bernstein, ``{Almost Eigenvectors for Almost Commuting Matrices},'' {\em
  \href{http://dx.doi.org/10.1137/0121026}{SIAM Journal on Applied
  Mathematics}} {\bfseries 21}, 232--235,  (1971).

\bibitem{BridgemanFlammiaPoulin2016}
J.~C. Bridgeman, S.~T. Flammia, and D.~Poulin, ``{Detecting Topological Order
  with Ribbon Operators},''
  \href{http://arxiv.org/abs/1603.02275}{arXiv:1603.02275},  (2016).

\bibitem{Bhatia}
R.~Bhatia, {\em \href{http://dx.doi.org/10.1007/978-1-4612-0653-8}{{\em {Matrix
  Analysis}}}}, ~169 of {\em Graduate Texts in Mathematics}.
\newblock Springer New York,  (1997).

\bibitem{Schatten}
R.~Schatten, {\em \href{http://dx.doi.org/10.1007/978-3-642-87652-3}{{\em {Norm
  Ideals of Completely Continuous Operators}}}}.
\newblock Springer Berlin, Heidelberg,  (1960).

\bibitem{vonNeumann}
J.~{Von Neuman}, ``{Some matrix inequalities and metrization of matric
  spaces},'' {\em Tomsk Univ. Rev.} {\bfseries 1}, 286--299,  (1937).

\bibitem{HoffmanWielandt1953}
A.~J. Hoffman and H.~W. Wielandt, ``{The variation of the spectrum of a normal
  matrix},'' {\em \href{http://dx.doi.org/10.1215/S0012-7094-53-02004-3}{Duke
  Mathematical Journal}} {\bfseries 20}, 37--39,  (1953).

\bibitem{Bravyi2010a}
S.~Bravyi, M.~B. Hastings, and S.~Michalakis, ``{Topological quantum order:
  Stability under local perturbations},'' {\em
  \href{http://dx.doi.org/10.1063/1.3490195}{Journal of Mathematical Physics}}
  {\bfseries 51}, 093512,
  \href{http://arxiv.org/abs/1001.0344}{arXiv:1001.0344},  (2010).

\bibitem{Hastings2005}
M.~B. Hastings and X.-G. Wen, ``{Quasiadiabatic continuation of quantum states:
  The stability of topological ground-state degeneracy and emergent gauge
  invariance},'' {\em
  \href{http://dx.doi.org/10.1103/PhysRevB.72.045141}{Physical Review B}}
  {\bfseries 72}, 045141,
  \href{http://arxiv.org/abs/cond-mat/0503554}{arXiv:cond-mat/0503554},
  (2005).

\bibitem{Kitaev2003}
A.~Kitaev, ``{Fault-tolerant quantum computation by anyons},'' {\em
  \href{http://dx.doi.org/10.1016/S0003-4916(02)00018-0}{Annals of Physics}}
  {\bfseries 303}, 2--30,
  \href{http://arxiv.org/abs/quant-ph/9707021}{arXiv:quant-ph/9707021},
  (2003).

\bibitem{Hu2012}
Y.~Hu, Y.~Wan, and Y.-S. Wu, ``{Twisted Quantum Double Model of Topological
  Phases in Two--Dimension},'' {\em
  \href{http://dx.doi.org/10.1103/PhysRevB.87.125114}{Physical Review B}}
  {\bfseries 87}, 125114,
  \href{http://arxiv.org/abs/1211.3695}{arXiv:1211.3695},  (2012).

\bibitem{Dijkgraaf1990}
R.~Dijkgraaf and E.~Witten, ``{Topological gauge theories and group
  cohomology},'' {\em
  \href{http://dx.doi.org/10.1007/BF02096988}{Communications in Mathematical
  Physics}} {\bfseries 129}, 393--429,  (1990).

\bibitem{LevinWen2005}
M.~A. Levin and X.-G. Wen, ``{String-net condensation: A physical mechanism for
  topological phases},'' {\em
  \href{http://dx.doi.org/10.1103/PhysRevB.71.045110}{Physical Review B}}
  {\bfseries 02139},
  \href{http://arxiv.org/abs/cond-mat/0404617}{arXiv:cond-mat/0404617},
  (2004).

\bibitem{Haah2014}
J.~Haah, ``{An Invariant of Topologically Ordered States Under Local Unitary
  Transformations},'' {\em
  \href{http://dx.doi.org/10.1007/s00220-016-2594-y}{Communications in
  Mathematical Physics}} {\bfseries 342}, 771--801,
  \href{http://arxiv.org/abs/1407.2926}{arXiv:1407.2926},  (2016).

\bibitem{BridgemanChubb2016}
J.~C. Bridgeman and C.~T. Chubb, ``{Hand-waving and Interpretive Dance: An
  Introductory Course on Tensor Networks},''
  \href{http://arxiv.org/abs/1603.03039}{arXiv:1603.03039},  (2016).

\bibitem{Orus2014}
R.~Or{\'{u}}s, ``{A practical introduction to tensor networks: Matrix product
  states and projected entangled pair states},'' {\em
  \href{http://dx.doi.org/10.1016/j.aop.2014.06.013}{Annals of Physics}}
  {\bfseries 349}, 117--158,
  \href{http://arxiv.org/abs/1306.2164}{arXiv:1306.2164},  (2014).

\bibitem{NielsenChuang2011}
M.~A. Nielsen and I.~L. Chuang, {\em {\em {Quantum Computation and Quantum
  Information 10th Anniversary Edition}}}.
\newblock Cambridge University Press, Cambridge,  (2011).

\bibitem{Knill1996}
E.~Knill, R.~Laflamme, and L.~Viola, ``{Theory of Quantum Error Correction for
  General Noise},'' {\em
  \href{http://dx.doi.org/10.1103/PhysRevLett.84.2525}{Physical Review
  Letters}} {\bfseries 84}, 2525--2528,
  \href{http://arxiv.org/abs/quant-ph/9604034}{arXiv:quant-ph/9604034},
  (2000).

\bibitem{Blume-Kohout2010}
R.~Blume-Kohout, H.~K. Ng, D.~Poulin, and L.~Viola, ``{Information-preserving
  structures: A general framework for quantum zero-error information},'' {\em
  \href{http://dx.doi.org/10.1103/PhysRevA.82.062306}{Physical Review A}}
  {\bfseries 82}, 062306,
  \href{http://arxiv.org/abs/1006.1358}{arXiv:1006.1358},  (2010).

\bibitem{Glashoff2013}
K.~Glashoff and M.~M. Bronstein, ``{Almost-commuting matrices are almost
  jointly diagonalizable},''
  \href{http://arxiv.org/abs/1305.2135}{arXiv:1305.2135},  (2013).

\bibitem{Wannier}
F.~Gygi, J.-L. Fattebert, and E.~Schwegler, ``{Computation of Maximally
  Localized Wannier Functions using a simultaneous diagonalization
  algorithm},'' {\em
  \href{http://dx.doi.org/10.1016/S0010-4655(03)00315-1}{Computer Physics
  Communications}} {\bfseries 155}, 1--6,  (2003).



\bibitem{geometricmean}
M.~Congedo, B.~Afsari, A.~Barachant, and M.~Moakher, ``{Approximate Joint Diagonalization and Geometric Mean of Symmetric Positive Definite Matrices},'' 
{\em \href{http://dx.doi.org/10.1371/journal.pone.0121423}{Public Library of Science}}, {\bfseries 10}, 4, (2015).




\bibitem{Laplacian}
D.~Eynard, A.~Kovnatsky, and M.~{M. Bronstein}, ``{Laplacian colormaps: a
  framework for structure-preserving color transformations},'' {\em
  \href{http://dx.doi.org/10.1111/cgf.12295}{Computer Graphics Forum}}
  {\bfseries 33}, 215--224,
  \href{http://arxiv.org/abs/1311.0119}{arXiv:1311.0119},  (2014).

\bibitem{nearby}
C.~Pearcy and A.~Shields, ``{Almost commuting matrices},'' {\em
  \href{http://dx.doi.org/10.1016/0022-1236(79)90071-5}{Journal of Functional
  Analysis}} {\bfseries 33}, 332--338,  (1979).

\bibitem{combin}
S.~H. Schanuel, ``{A combinatorial problem of Shields and Pearcy},'' {\em
  \href{http://dx.doi.org/10.1090/S0002-9939-1977-0439652-6}{Proceedings of the
  American Mathematical Society}} {\bfseries 65}, 185--185,  (1977).

\bibitem{GyarfasLehel1985}
A.~Gy{\'{a}}rf{\'{a}}s and J.~Lehel, ``{Covering and coloring problems for
  relatives of intervals},'' {\em
  \href{http://dx.doi.org/10.1016/0012-365X(85)90045-7}{Discrete Mathematics}}
  {\bfseries 55}, 167--180,  (1985).

\end{thebibliography}
\normalsize
\appendix


\section{Approximate shared eigenvectors for \add{approximately}\del{almost} commuting matrices}
\label{app:approxeigen}

In this section we will show that for two approximately commuting matrices, an approximate shared eigenvector exists. This problem has been considered before by Bernstein~\cite{Bernstein1971}, who showed the following result.

\begin{thm}[\!\cite{Bernstein1971}]
Take $A$ and $B$ to be complex matrices of dimension $n\ge 2$. If $\norm{B}\leq 1$, and for some $\delta>0$ we have
\[ \norm{\comm{A}{B}}\leq  \frac{\delta^n(1-\delta)}{1-\delta^{n-1}}, \]
then for each eigenvalue $\lambda$ of $A$, there exists a $\mu$ and normalized $\ket{x}$ such that
\[ \norm{A\ket{x}-\lambda\ket{x}},\norm{B\ket{x}-\mu\ket{x}}\leq\delta. \]
\end{thm}

Notice above the required bound on the commutator \del{for a given $\delta$ 
}scales like $\add{\mathcal{O}(\delta^n)}$ \add{for small $\delta$}. Below we 
will improve this dimension 
scaling 
by adding the additional assumption that one of the matrices is normal, 
allowing us to bring this down to a $\add{\mathcal{O}(\delta^2/n^2)}$ 
dependence. First 
we 
will state the more general result, which only requires one of the matrices to 
be normal\add{, followed by a more specialized result which applies when both 
matrices are normal.}\del{. After proving this, we will then give the some 
additional cases where the second matrix is either unitary or Hermitian and the 
output is required to satisfy some additional constraints.}

\add{The existence of an entire basis of shared approximate eigenvectors is 
closely related to approximate joint diagonalization, a problem that has been 
widely considered
and has found application in fields such as quantum chemistry~\cite{Wannier}, 
machine learning~\cite{geometricmean}  
and 
image processing~\cite{Laplacian}. This literature is too vast to review in 
this appendix, but see Ref.~\cite{Glashoff2013} for a discussion of the 
relationship 
between approximately commuting matrices and joint diagonalization. Techniques 
similar to those used below have 
also been used in Ref.~\cite{nearby} to address the related problem of 
constructing nearby exactly commuting operators, in the case in which one 
matrix is Hermitian. Whilst this analysis gives better bounds than those 
presented below, it leverages a combinatorical construction~\cite{combin} that 
explicitly uses the reality of the eigenvalues, and therefore cannot be 
directly applied to the case we will consider in which one matrix is normal, 
but not necessarily Hermitian.}

\del{Before we prove this, we will need two simple lemmas. }Take $A$ and $B$ to be $n\times n$ matrices. Let $A$ be normal, with an eigenvalue decomposition $A=\sum_{i}\lambda_i\ket{i}\bra{i}$. \del{Given this, we will see that we can upper bound the off-diagonal terms of $B$ in this eigenbasis of $A$.} Next take $\lambda$ to be a specific eigenvalue of $A$. Let $I_0$ be the singleton set containing the index corresponding to $\lambda$, or all these indices if $\lambda$ is degenerate. Define $I_k$ to be all the indices whose eigenvalues are within some radius $r > 0$ in the complex plane (to be chosen later) of those in $I_{k-1}$, i.e.\
\[ I_k:=\left\lbrace i \,\middle|\, \exists j\in I_{k-1} : \abs{\lambda_i-\lambda_j}\leq r \right\rbrace. \]
Clearly this sequence becomes fixed after at most $n$ terms, and so let $I:=I_n$ be this fixed point. Intuitively $I$ can be thought of as the indices corresponding to eigenvalues which form a cluster around $\lambda$ where every eigenvalue in the cluster is linked to at least one other by a disk of radius $r$ in the complex plane. 

By construction this set has two properties we require. First it is bounded away from any other index,
\[ i\in I, j\notin I \implies \abs{\lambda_i-\lambda_j} > r\]
Second, because all of the eigenvalues corresponding to elements in $I$ have nearby neighbors in $I$, this means that the diameter of the disk containing all of the eigenvalues in $I$ has a diameter bounded by at most $nr$, 
\[ i\in I\implies \abs{\lambda_i-\lambda}\leq nr \,.\]

Next let $V$ be the space spanned by the eigenvectors whose indices lies in $I$,
\[ V:=\Span \left\lbrace\ket{i}\middle|\,i\in I\right\rbrace. \]
Denot\add{ing}\del{e} the orthogonal complement of $V$ by $\bar{V}$, \add{then we can}\del{and} decompose both $A$ and $B$ into blocks \add{on}\del{by} $V\oplus \bar V$ as
\[ A=\begin{pmatrix}
A_V & \\ & A_{\bar V}
\end{pmatrix}
\qquad\text{and}\qquad 
B=\begin{pmatrix}
B_{VV} &B_{\bar{V}V} \\B_{V\bar V} & B_{\bar V\bar V}
\end{pmatrix}. \]

\del{\begin{lem}[Off-diagonal bound]
	\label{lem:off-diag2}
	If $\norm{\comm{A}{B}}\leq \epsilon$ then 
	\[\abs{\lambda_i-\lambda_j}\cdot\abs{\braopket{i}{B}{j}}\leq \epsilon. \]
\end{lem}
\begin{proof}
	Using the fact that the operator norm dominates any component of a matrix, we can simply evaluate the relevant component of the commutator:
	\begin{align*}
	\epsilon
	&\geq \bigl\|\comm{A}{B}\bigr\|\\
	&\geq \abs{\bra{i}\comm{A}{B}\ket{j}}\\
	&= \abs{\bra{i}\left[AB-BA\right]\ket{j}}\\
	&=\abs{\lambda_i-\lambda_j}\cdot\abs{\braopket{i}{B}{j}}\,.
	\end{align*}
\end{proof}}

\begin{lem}
\label{lem:strictineq}
	\add{If $\norm{\comm{A}{B}}\leq \epsilon$, with $A$ normal and decomposed 
	as above, then $A_V$ is close to scalar, 
	and the off-diagonal blocks of $B$ are bounded as
	\[ \norm{A_V-\lambda \mathbbm{1}_V}\leq nr\qquad \text{and}\qquad 
	\norm{B_{\bar 
	VV}}\leq n\epsilon/2. \]
	}
	\del{The off-diagonal blocks of $B$ are bounded
	\[ \norm{B_{\bar VV}} < n\epsilon/2r. \]}
\end{lem}
\begin{proof}
	\add{Given that $A$ is normal, we can see that $A_V$ is approximately 
	scalar due to the bound between eigenvalues in $I$:
	\begin{align*}
		\norm{A_V-\lambda\mathbbm 1_V}=\max_{i\in I}\abs{\lambda_i-\lambda}\leq 
		nr.
	\end{align*}}	
	\add{\!\!Next, using the fact that the operator norm dominates any 
	component of a matrix, we can simply evaluate the relevant component of the 
	commutator to bound elements of $B$:
	\begin{align*}
	\epsilon
	&\geq \bigl\|\comm{A}{B}\bigr\|\\
	&\geq \abs{\bra{i}\comm{A}{B}\ket{j}}\\
	&= \abs{\bra{i}\left[AB-BA\right]\ket{j}}\\
	&=\abs{\lambda_i-\lambda_j}\cdot\abs{\braopket{i}{B}{j}}\,.
	\end{align*}
	\!\!}
	In the eigenbasis of $A$, the components of $B_{\bar VV}$ correspond to $\braopket{i}{B}{j}$ for $i\notin I$, $j\in I$. By construction of $I$ we have that $\abs{\lambda_i-\lambda_j}> r$, and so \del{by \cref{lem:off-diag2}} 
	\[ \abs{\braopket{i}{B}{j}}\leq \frac{\epsilon}{\abs{\lambda_i-\lambda_j}} < \frac{\epsilon}{r}. \]
	This implies therefore that $\norm{B_{\bar VV}}_\text{max} < \epsilon/r$, where $\norm{\cdot}_{\max}$ denotes the elementwise max-norm. Using the fact that the operator norm exceeds the max-norm by at most the square root of the number of elements, we get
	\[ \norm{B_{\bar VV}}\leq \norm{B_{\bar VV}}_\text{max}\times\sqrt{\dim V\times \dim \bar{V}}. \]
	Given that $\dim V+\dim \bar{V}=n$, we have that $\dim V\times \dim \bar{V}\leq n^2/4$, and so
	\[ \norm{B_{\bar VV}}< n\epsilon/2r\,. \]
\end{proof}

\add{Using these bounds, we can now put bounds on an approximate shared eigenvector. Imposing normality on both matrices, we can even impose the stricter requirement that both of the approximate eigenvalues are in fact exact eigenvalues.}

\begin{appthm}[\add{Shared approximate eigenvector}\del{Fixed eigenvalue}]
	\label{appthm:fixed}
	Suppose that $A$ and $B$ are \add{$n\times n$}\del{$n$-dimensional} matrices, such that $A$ is normal and $\norm{\comm{A}{B}}\leq \epsilon$. For \add{any}\del{a fixed} $\lambda$ which is an eigenvalue of $A$, there exists a normalized $\ket{u}$ and $\mu$ such that
	\[ \norm{A\ket{u}-\lambda\ket{u}},\norm{B\ket{u}-\mu\ket{u}}\leq n\sqrt{\epsilon/2}. \]
	\add{If $B$ is also normal, then for any $\lambda$ which is an eigenvalue 
	of $A$, there exists a $\nu$ which is also an eigenvalue of $B$ and 
	normalized $\ket{w}$ such that
		\[ \norm{A\ket{w}-\lambda\ket{w}},
	\norm{B\ket{w}-\nu\ket{w}}\leq n\sqrt{\epsilon}. \]
	}
\end{appthm}\begin{proof}
	Take $\ket{u}$ to be a right eigenvector of $B_{VV}$ (contained within $V$), of eigenvalue $\mu$. \add{By \cref{lem:strictineq}, t}\del{T}his then gives that the relevant errors with respect to $A$ and $B$ behave as:
	
	\begin{align*}
	\norm{A\ket{u}-\lambda\ket{u}}&=\norm{A_V\ket{u}-\lambda \ket{u}} & 
	\norm{B\ket{u}-\mu\ket{u}}&=\norm{B_{VV}\ket{u}-\mu \ket{u}+B_{\bar VV}\ket{u}}\\
	&=\norm{(A_V-\lambda\mathbbm 1)\ket{u}} & &=\norm{B_{\bar VV}\ket{u}}\\
	&\leq\norm{A_V-\lambda\mathbbm 1} & &\leq\norm{B_{\bar VV}}\\ 
	&\leq nr & &< n\epsilon/2r \,.
	\end{align*}
	\del{Here the inequality for the $A$ matrix comes from the diameter bound on the disk containing the eigenvalues in $I$, and the inequality for the $B$ matrix comes from Lemma~\ref{lem:strictineq}.
	}
	\del{Next we can pick an $r$ for which the maximum of these two terms is minimized, specifically $r=\sqrt{\epsilon/2}$. Substituting this into each bound gives}\add{If we now let $r=\sqrt{\epsilon/2}$, we get} the \add{stated }overall bound of $n\sqrt{\epsilon/2}$.
	
	\add{For the case of both matrices being normal, we can show that any approximate eigenvalue must lie near an exact eigenvalue. Taking $\ket{w}$ once again to be a right eigenvector of $B_{VV}$ with eigenvalue $\nu'$ (for a different value of $r$ to $\ket{u}$), we can see that
	\begin{align*}
		\norm{B\ket{w}-\nu'\ket{w}}\leq n\epsilon/2r\qquad\implies\qquad \bra{w}(B-\nu')^\dag (B-\nu)\ket{w}\leq n^2\epsilon^2/4r^2.
	\end{align*}
	As $(B-\nu')^\dag(B-\nu')$ is positive semi-definite, the existence of 
	such a $\ket{w}$ implies $(B-\nu')^\dag (B-\nu')$ possesses an eigenvalue 
	at most $n^2\epsilon^2/4r^2$. By the normality of $B$, this implies in turn 
	that $B$ contains an eigenvalue $\nu$ such that $\abs{\nu-\nu'}\leq 
	n\epsilon/2r$. Using this we can see that the error with respect to $B$ 
	gains a factor of $2$
	\begin{align*}
		\norm{B\ket{w}-\nu\ket{w}}\leq \norm{B\ket{w}-\nu'\ket{w}}+\abs{\nu-\nu'}\leq n\epsilon/r.
	\end{align*}
	Now taking $r=\sqrt{\epsilon}$, we find the stated bound of $n\sqrt{\epsilon}$.
	}
\end{proof}

\del{
In general $\mu$ is not an eigenvalue of $B$. We can impose that this value have unit norm in the unitary case, or impose it to be an eigenvalue of $B$ in the Hermitian case, at the additional cost of a factor of $\sqrt{2}$.
\begin{corr}[Unitary case]
	\label{corr:unit}
	Suppose that $A$ and $B$ are $n$-dimensional matrices, such that $A$ is normal, $B$ is unitary, and $\norm{\comm{A}{B}}\leq \epsilon$. For a fixed $\lambda$ which is an eigenvalue of $A$, there exists a vector $\ket{u}$ and $\mu$ with $\abs{\mu}=1$, such that
	\[ \norm{A\ket{u}-\lambda\ket{u}},\norm{B\ket{u}-\mu\ket{u}}\leq n\sqrt{\epsilon}. \]
\end{corr}
\begin{proof}
	The proof is similar to above. Suppose we took $\ket{u}$ to be a right eigenvector of $B_{VV}$ with eigenvalue $\mu'$, which since $B$ is unitary implies that $|\mu'| \le 1$. Let $\mu=\mu'/\abs{\mu'}$ such that $\abs{\mu-\mu'}=1-\abs{\mu'}$. Requiring that $B$ is a unitary matrix implies
	\[ B_{VV}^\dag B_{VV}+B_{\bar VV}^\dag B_{\bar VV}=I_{V}. \]
	As $\ket{u}$ is a $\mu'$-eigenvector of $B_{VV}$, we can upper bound the difference $\abs{\mu-\mu'}$:
	\begin{align*}
		\abs{\mu-\mu'}
		&=1-\abs{\mu'}\\
		&\leq 1-\abs{\mu'}^2\\
		&=\bigl\langle u\bigr| \bigl[I_V-B_{VV}^\dag B_{VV}\bigr]\bigl| u\bigr\rangle\\
		&\leq\bigl\| I_V-B_{VV}^\dag B_{VV}\bigr\|\\
		&=\bigl\| B_{\bar VV}^\dag B_{\bar VV}\bigr\|\\
		&=\bigl\|B_{\bar VV}\bigr\|^2\\
		&\leq\norm{B_{\bar VV}}\\
		&< n\epsilon/2r \,.
	\end{align*}
	The cost associated with $B$ now becomes
	\begin{align*}
		\norm{B\ket{u}-\mu\ket{u}}
		&\leq\norm{B\ket{u}-\mu'\ket{u}}+\abs{\mu-\mu'}\\
		&< n\epsilon/2r+n\epsilon/2r\\
		&=n\epsilon/r \,.
	\end{align*}
	The bound of $\norm{A\ket{u}-\lambda\ket{u}} \le n r$ carries over unchanged. Once again minimaxing over $r$ ($r=\sqrt{\epsilon}$), we get the stated bound.
\end{proof}
\begin{corr}[Hermitian case]
	\label{corr:herm}
	Suppose that $A$ and $B$ are $n$-dimensional matrices, such that $A$ is normal, $B$ is Hermitian, and $\norm{\comm{A}{B}}\leq \epsilon$. For a fixed $\lambda$ which is an eigenvalue of $A$, there exists a vector $\ket{u}$ and $\mu$ which is an eigenvalue of $B$, such that
	\[ \norm{A\ket{u}-\lambda\ket{u}},\norm{B\ket{u}-\mu\ket{u}}\leq n \sqrt{\epsilon}. \]
\end{corr}
\begin{proof}
	Once again, let $\ket{u}$ be an eigenvector of $B_{VV}$ of value $\mu'$. Next let $B'$ be the pinching of $B$ to $V\oplus \bar{V}$
	\[ B'=\begin{pmatrix}
	B_{VV} & 0 \\ 0 & B_{\bar V\bar V}
	\end{pmatrix}. \]
	First we note that $B$ and $B'$ are close in the operator norm, as
	\begin{align*}
		\norm{B-B'}
		&=\norm{B_{\bar VV}+B_{V\bar V}}\\
		&=\max\lbrace \norm{B_{\bar VV}},\norm{B_{V\bar V}}\rbrace\\
		&\leq n\epsilon/2r \,.
	\end{align*}
	By definition $\mu'$ was a eigenvalue of $B_{VV}$, and therefore of $B'$. Weyl's inequality~\cite{Weyl1912} gives that there exists an eigenvalue $\mu$ of $B$ that is within the operator norm distance of $\mu'$, i.e.\
	\[ \abs{\mu-\mu'}\leq n\epsilon/2r. \]
	As with the unitary case, this $\abs{\mu-\mu'}$ also induces a factor of $2$ in the cost associated with $B$, and therefore gives the same final errors bounds.
\end{proof}
In all of the above analysis, we have considered $\lambda$ to be fixed. If instead we allow $\lambda$ to be \emph{any} eigenvalue of $A$, then we can get a factor of 2 improvement in our bounds.
\begin{appthm}[Free eigenvalue]
	\label{appthm:free}
	Suppose that $A$ and $B$ are $n$-dimensional matrices, such that $A$ is normal and $\norm{\comm{A}{B}}\leq \epsilon$. There exists a $\lambda$ which is an eigenvalue of $A$, and $\ket{u}$ and $\mu$ such that
	\[ \norm{A\ket{u}-\lambda \ket{u}},\norm{B\ket{u}-\mu\ket{u}}\leq n\sqrt{\epsilon}/2. \]
	Once again if $B$ is unitary and we require $\mu$ to have unit norm, or if $B$ is Hermitian and we require $\mu$ to be an eigenvalue, then the errors grow to
	\[ \norm{A\ket{u}-\lambda \ket{u}},\norm{B\ket{u}-\mu\ket{u}}\leq n\sqrt{\epsilon/2}. \]
\end{appthm}
\begin{proof}
	The proof here is once again similar to \cref{appthm:fixed}. Suppose we construct a set $I$ as before, from a cluster of eigenvalues with each element of $I$ within $r$ of at least one other element of $I$. The distance between any two eigenvalues with indices in $I$ is at most $nr$. There also exist eigenvalues with indices in $I$ such that all other such eigenvalues lie within $nr/2$ as follows. Consider the graph $G$ whose vertices are the points $\lambda_i \in I$ and where $(\lambda_i, \lambda_j)$ is an edge iff $|\lambda_i - \lambda_j| \le r$. This graph had diameter less than $n$ (the maximum number of point in $I$). Consider a maximum-length path in this graph, and choose the vertex corresponding to a midpoint of this path. In the complex plane, by the triangle inequality the eigenvalue associated to this midpoint vertex is at most distance $nr/2$ away from the other eigenvalues.
	Let $\lambda$ be such a midpoint eigenvalue. This choice allows us to improve the error bound on the $A$ approximate eigenvector by a factor of two,
	\[ \norm{A\ket{u}-\lambda \ket{u}}\leq nr/2 \,.\]
	Here the analysis works the same as previously. This gives an overall factor of $\sqrt{2}$ improvement,
	\[ \norm{A\ket{u}-\lambda \ket{u}},\norm{B\ket{u}-\mu\ket{u}}\leq n\sqrt{\epsilon}/2\,. \] 	
	Using the same techniques as Corollaries~\ref{corr:unit}~and~\ref{corr:herm} we can impose $\mu$ to be either unit norm or an eigenvalue respectively, growing the error associated with $B$ to
	\[ \norm{B\ket{u}-\mu\ket{u}}\leq n\epsilon/r. \]
	This factor of two carries through to a $\sqrt{2}$ in the final error, giving the stated final bound of $n\sqrt{\epsilon/2}$.
\end{proof}
}


\section{An algorithm for the certifiable degeneracy of a twisted pair}
\label{app:lowerbound}

In this appendix we sketch how, for a given pair of parameters $\alpha$ and $\delta$, we can calculate the minimum possible dimension of unitaries $u$ and $v$ such that $\normu{\twicomm{u}{v}}\leq\delta.$ Lemmas~\ref{lem:t-col:change1}~and~\ref{lem:t-col:eigen1} give that for all $j\in \mathbb{Z}$, there exists an eigenvalue $e^{i\phi_j}$ of $u$ such that
\[ \abs{\phi_j-2\pi \alpha j}\leq \cos^{-1}(1-\abs{j}\delta). \]
The question now is to find the minimum number of eigenvalues such that at least one lies in each of the above arcs. This is known as the \emph{transversal number}, and can be efficiently calculated by a greedy algorithm~\cite{GyarfasLehel1985}. We now sketch this algorithm for the example parameters $\alpha=1/4$ and $\delta=1/2$ (indicated by the turquoise dot in \cref{fig:mountains}).

The first thing to note is that these arcs are trivial for $\delta \abs{j}\geq 2$, in that they are the entire unit circle. For this reason we need only consider a finite number of arcs for $j=-\lfloor 2/\delta\rfloor,\ldots,\lfloor 2/\delta\rfloor$. In our case this corresponds $j=-3,\dots,3$. Below we have drawn these non-trivial arcs, omitting the trivial $j=0$ arc.

\begin{center}
	\ifcompile
	\tikzsetnextfilename{Arcs1}
	\begin{tikzpicture}[scale=0.75]
	\def\dr{0.125}
	\draw[very thick] (0,0) circle (1);
	\draw[very thick] (0,0) circle (1+11*\dr);
	
	\def\L{30} \def\U{150}
	\draw[fill=blue] (\L:1+\dr) arc (\L:\U:1+\dr) -- (\U:1+2*\dr) arc (\U:\L:1+2*\dr) -- cycle;
	\draw[fill=blue] (-\L:1+3*\dr) arc (-\L:-\U:1+3*\dr) -- (-\U:1+4*\dr) arc (-\U:-\L:1+4*\dr) -- cycle;
	
	\def\L{90} \def\U{270}
	\draw[fill=blue] (\L:1+5*\dr) arc (\L:\U:1+5*\dr) -- (\U:1+6*\dr) arc (\U:\L:1+6*\dr) -- cycle;
	
	\def\L{150} \def\U{390}
	\draw[fill=blue] (\L:1+7*\dr) arc (\L:\U:1+7*\dr) -- (\U:1+8*\dr) arc (\U:\L:1+8*\dr) -- cycle;
	\draw[fill=blue] (-\L:1+9*\dr) arc (-\L:-\U:1+9*\dr) -- (-\U:1+10*\dr) arc (-\U:-\L:1+10*\dr) -- cycle;
	
	\def\a{0}
	\end{tikzpicture}
	\else
	\includegraphics{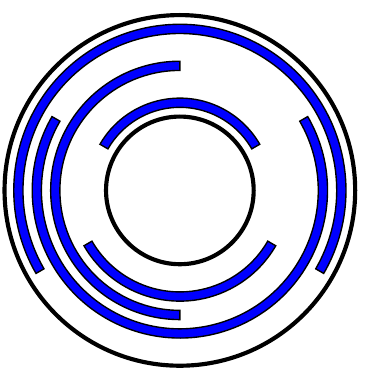}
	\fi
\end{center}

~

Next we note that the $j=0$ arc is simply a point, implying that $u$ must contain a $+1$ eigenvalue. Given this, any arc containing $+1$ can be thrown away (indicated in red below), allowing us to unfold our arcs on a circle into intervals on a line.

\begin{center}
	\ifcompile
	\tikzsetnextfilename{Arcs2}
	\begin{tikzpicture}[scale=0.75]
	\def\dr{0.125}
	\draw[very thick] (0,0) circle (1);
	\draw[very thick] (0,0) circle (1+11*\dr);
	
	\def\L{30} \def\U{150}
	\draw[fill=blue] (\L:1+\dr) arc (\L:\U:1+\dr) -- (\U:1+2*\dr) arc (\U:\L:1+2*\dr) -- cycle;
	\draw[fill=blue] (-\L:1+3*\dr) arc (-\L:-\U:1+3*\dr) -- (-\U:1+4*\dr) arc (-\U:-\L:1+4*\dr) -- cycle;
	
	\def\L{90} \def\U{270}
	\draw[fill=blue] (\L:1+5*\dr) arc (\L:\U:1+5*\dr) -- (\U:1+6*\dr) arc (\U:\L:1+6*\dr) -- cycle;
	
	\def\L{150} \def\U{390}
	\draw[fill=red] (\L:1+7*\dr) arc (\L:\U:1+7*\dr) -- (\U:1+8*\dr) arc (\U:\L:1+8*\dr) -- cycle;
	\draw[fill=red] (-\L:1+9*\dr) arc (-\L:-\U:1+9*\dr) -- (-\U:1+10*\dr) arc (-\U:-\L:1+10*\dr) -- cycle;
	
	\def\a{0}
	\draw[black,very thick] (\a:1) -- (\a:1+11*\dr);
	\draw[very thick] (3.25,0) -- (5,0) -- (4.8,0.2)  (4.8,-0.2) -- (5,0);	
	\begin{scope}[xshift=6cm,xscale=1/90]
	\draw[very thick] (0,0) -- (360,0);
	\draw[very thick] (0,7*\dr) -- (360,7*\dr);
	\draw[fill=blue] (30,\dr) rectangle (150,2*\dr);
	\draw[fill=blue] (90,3*\dr) rectangle (270,4*\dr);
	\draw[fill=blue] (330,5*\dr) rectangle (210,6*\dr);
	\end{scope}
	\end{tikzpicture}
	\else
	\includegraphics{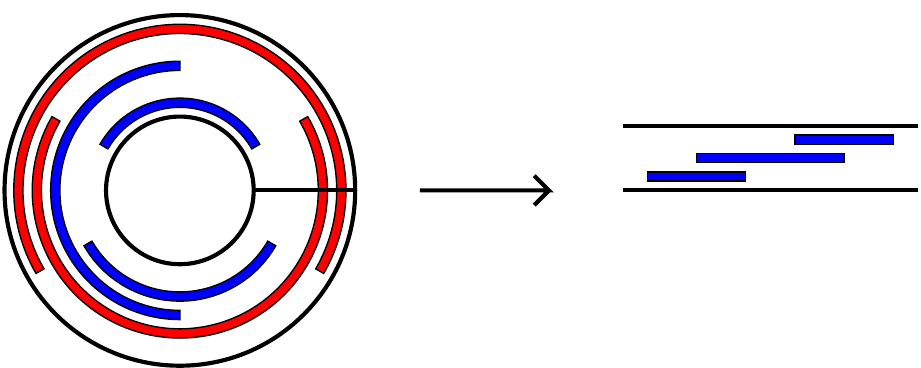}
	\fi
\end{center}

~

We then take the intervals to be sorted by end-point. Considering each interval in order, we place an eigenvalue at the end of each interval as necessary, indicated as a green line below. Note that any interval which already contains an included eigenvalue when we arrive at it can be ignored, indicated by the red interval below. 

\begin{center}
	\ifcompile
	\tikzsetnextfilename{Interval}
	\begin{tikzpicture}
		\def\dr{0.125}
		\begin{scope}[xscale=1/90]
			\draw[very thick] (0,0) -- (360,0);
			\draw[very thick] (0,7*\dr) -- (360,7*\dr);
			\draw[fill=blue] (30,\dr) rectangle (150,2*\dr);
			\draw[fill=red] (90,3*\dr) rectangle (270,4*\dr);
			\draw[fill=blue] (330,5*\dr) rectangle (210,6*\dr);
			\draw[green, very thick] 
			(150,0) -- (150,7*\dr) (330,0) -- (330,7*\dr);
		\end{scope}
	\end{tikzpicture}
	\else
	\includegraphics{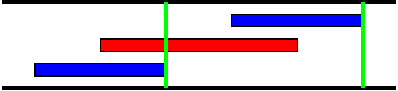}
	\fi
\end{center}
Including the already found eigenvalue at $+1$, this gives us the minimum number of points necessary to satisfy each arc. Applying the algorithm for a large number of points, we can plot the certified degeneracy as in \cref{fig:mountains}.


\end{document}